\documentclass[11pt,oneside]{article} 
\pdfoutput=1
\usepackage{geometry}   
\usepackage{graphicx}   
\usepackage{amssymb}
\usepackage{float}
\usepackage[cmex10]{amsmath}
\usepackage{amsthm}

\newcommand\bbR{\mathbb R}
\newcommand\bx{\boldsymbol x}
\newcommand\by{\boldsymbol y}
\newcommand\bE{\boldsymbol E}
\newcommand\ds{\boldsymbol ds}
\newcommand\bH{\boldsymbol H}
\newcommand\bA{\boldsymbol A}
\newcommand\bJ{\boldsymbol J}

\newcommand\In{\operatorname{inc}}
\newcommand\Sc{\operatorname{scat}}
\newcommand\bn{\boldsymbol n}

\newcommand\tot{\operatorname{}}
\newtheorem{remark}{Remark}
\newtheorem{definition}{Definition}
\newtheorem{thm}{Theorem}
\newtheorem{lemma}{Lemma}

\title{The Decoupled Potential Integral Equation for Time-Harmonic 
Electromagnetic Scattering}
\author{Felipe Vico\thanks{Instituto de Telecomunicaciones y Aplicaciones Multimedia (ITEAM), 
Universidad Polit\` ecnica
de Val\` encia, 46022 Val\` encia, Spain. {{\em email}: {\sf {felipe.vico@gmail.com,
}}} {{\sf {mferrand@dcom.upv.es.
}}} } \and
Leslie Greengard\thanks{Courant Institute of Mathematical Sciences,
         New York University, 
         251 Mercer Street,
         New York, NY 10012-1110.
{{\em email}: {\sf {greengard@cims.nyu.edu.}}}} \and
Miguel Ferrando\footnotemark[1]\and
 Zydrunas Gimbutas\thanks{%
Information Technology Laboratory,
National Institute of Standards and Technology,
325 Broadway, Mail Stop 891.01,
Boulder, CO  80305-3328.
{{\em email}: {\sf {zydrunas.gimbutas@nist.gov}}. 
Contributions by staff of NIST, an agency of the U.S. Government, 
are not subject to copyright within the United States.}
}}

\begin{document}
\maketitle
\begin{abstract}
\bibliographystyle{unsrt}

We present a new formulation for the problem of electromagnetic scattering from
perfect electric conductors. While our representation for the electric and
magnetic fields is based on the standard vector and scalar potentials $\bA,\phi$
in the Lorenz gauge, we establish boundary conditions on the potentials themselves, 
rather than on the field quantities. This permits the development of a well-conditioned
second kind Fredholm integral equation which has no spurious resonances, 
avoids low frequency breakdown, and is insensitive to the genus of the
scatterer. The equations for the vector and scalar potentials are decoupled.
That is, the unknown scalar potential defining the scattered field, $\phi^{\Sc}$, 
is determined entirely by the incident scalar potential $\phi^{\In}$.
Likewise, the unknown vector potential defining the scattered field,
$\bA^{\Sc}$, is determined entirely by the incident vector potential $\bA^{\In}$. This decoupled formulation is valid not only in the static limit but for arbitrary $\omega\ge 0$.
\end{abstract}


\newpage
\section{Introduction}

In this paper, we consider the problem of exterior scattering of time-harmonic 
electromagnetic waves by perfect electric conductors. For
a fixed frequency $\omega$, we assume that the electric and magnetic fields take the form
\begin{equation}
\begin{aligned}
\mathcal{E}(\bx,t)&=\Re\big\{\bE(\bx)e^{-i\omega t}\big\},\\
\mathcal{H}(\bx,t)&=\Re\big\{ \bH(\bx)e^{-i\omega t}\big\}, 
\end{aligned}
\end{equation}
so that Maxwell's equations are
\begin{equation}
\begin{aligned}
\nabla\times\bE^{\tot}(\bx)&=i\omega\mu\bH^{\tot}(\bx),\\
\nabla\times\bH^{\tot}(\bx)&=-i\omega\epsilon\bE^{\tot}(\bx).
\end{aligned}
\end{equation}
Following standard practice, we write the {\em total} 
electric and magnetic fields as a sum 
of the (known) incident and (unknown) scattered fields:
\begin{equation}
\begin{aligned}
\bE &=\bE^{\In}+\bE^{\Sc},\\ 
\bH &=\bH^{\In}+\bH^{\Sc}.
\end{aligned}
\end{equation}
The scattered field in the exterior must satisfy the Sommerfeld-Silver-M\"{u}ller 
radiation condition:
\begin{equation}\label{EMrad}
\begin{array}{ll}
\bH^{\Sc}(\bx)\times\frac{\bx}{|\bx|}-\sqrt{\frac{\mu}{\epsilon}}\bE^{\Sc}(\bx)=o\Big(\frac{1}{|\bx|}\Big),&|\bx|\rightarrow \infty\, .
\end{array}
\end{equation}

It is well-known that when the scatterer, denoted by $D$, is a perfect
conductor, the conditions to be enforced on its boundary are
\cite{JACKSON,PAPAS}
\begin{eqnarray}\label{etang}
\bn\times\bE^{\tot}(\bx) &=& \mathbf{0}|_{\partial D},
\ \Rightarrow\ \bn\times\bE^{\Sc}(\bx)=-\bn\times\bE^{\In}(\bx)|_{\partial D},
\\
\label{hnorm}
\bn\cdot\bH^{\tot}(\bx) &=& 
0|_{\partial D},
\ \Rightarrow\ \bn\cdot\bH^{\Sc}(\bx)=-\bn\cdot\bH^{\In}(\bx)|_{\partial D},
\end{eqnarray}
where $\bn$ is the outward unit normal to the boundary $\partial D$ of the
scattered.  It is also well-known that
\begin{equation}\label{enorm}
\bn\cdot\bE^{\tot}(\bx)=\frac{\rho}{\epsilon}|_{\partial D},
\end{equation}
\begin{equation}\label{htang}
\bn\times\bH^{\tot}(\bx)=\bJ|_{\partial D},
\end{equation}
where $\bJ$ and $\rho$ are the induced current density and charge on the 
surface $\partial D$. In order to satisfy the Maxwell equations, 
$\bJ$ and $\rho$ must satisfy the continuity condition 
$\nabla_s\cdot \bJ=i\omega\rho$, where 
$\nabla_s \cdot \bJ$ denotes the surface divergence of the tangential current density.
It is also well-known that
the exterior problem for $\bE^{\Sc}$ 
has a unique solution for $\omega >0$  when boundary conditions are prescribed
on its tangential components (see, for example, \cite{CK2}):
\begin{equation}
\bn\times\bE^{\Sc}(\bx)=\mathbf{f(x)}|_{\partial D} \, ,
\end{equation}
for an arbitrary tangential vector field  $\mathbf{f}$. On a perfect conductor, 
$\mathbf{f(x)} = -\bn\times\bE^{\In}(\bx)$ to enforce (\ref{etang}).

\subsection{The vector and scalar potential}

Scattered electromagnetic fields are typically represented in terms of the
induced surface current $\bJ$ and charge $\rho$ using
the vector and scalar potentials in the Lorenz gauge:
\begin{eqnarray} \label{fields_potentials}
\bE^{\Sc} &=& i \omega \bA^{\Sc} - \nabla \phi^{\Sc},  \label{Epotrep} \\
\bH^{\Sc} &=& \frac{1}{\mu}\nabla \times \bA^{\Sc},  \label{Hpotrep} 
\end{eqnarray}
where
\[
\bA^{\Sc}[\bJ](\bx) =\mu S_k[\bJ](\bx) \equiv\mu  \int_{\partial D} g_k(\bx-\by) \bJ(\by) dA_{\by},  
\]
\begin{equation}
 \phi^{\Sc}[\rho](\bx) =   \frac{1}{\epsilon}S_k[\rho](\bx) \equiv
 \frac{1}{\epsilon}\int_{\partial D} g_k(\bx-\by) \rho(\by) dA_{\by}, 
\label{Skdef}
\end{equation}
with
\[ g_k(\bx) = \frac{e^{ik |\bx|}}{ 4 \pi | \bx|}  \]
and $k=\omega\sqrt{\epsilon\mu}$.
The Lorenz gauge is defined by the relation
\begin{equation}\label{Lorenz}
\nabla\cdot \bA^{\Sc}=i\omega\mu\epsilon\phi^{\Sc} \, .
\end{equation}
We will often refer to $\phi^{\Sc}$ and $\bA^{\Sc}$ as the {\em scalar
  and vector Helmholtz potentials} since $\phi^{\Sc}$ and $\bA^{\Sc}$
satisfy the Helmholtz equations with wavenumber $k$:
\begin{equation}
  \quad  \Delta \phi^{\Sc} +k^2 \phi^{\Sc} = 0, \quad
  \Delta \bA^{\Sc} +k^2 \bA^{\Sc} = \mathbf{0}.
\end{equation}
Using the representation (\ref{Epotrep}) for the electric field and imposing the 
boundary condition (\ref{etang}) results in the Electric Field Integral Equation (EFIE),   
\cite{CHEW,JIN,MAUE,Muller}:
\begin{align}
\label{efie}
 i \omega \bn \times  \bA^{\Sc}[\bJ](\bx) &- 
\bn \times\nabla \phi^{\Sc} \left[ \frac{\nabla_s \cdot \bJ}{i\omega} \right](\bx) \\
&= -\bn \times \bE^{\In}(\bx), \qquad  \bx \in \partial D. \nonumber
\end{align}
The representation (\ref{Hpotrep}) for the magnetic field and the 
boundary condition (\ref{htang}) results in the Magnetic Field Integral Equation (MFIE):
\begin{equation}
\frac{1}{2} \bJ(\bx) - K[\bJ](\bx)
= \bn(\bx) \times \bH^{\In}(\bx), \qquad  \bx \in \partial D,
\label{mfie}
\end{equation}
where 
\begin{equation}
K[\bJ](\bx) = 
 \int_{\partial D} \bn(\bx) \times \nabla \times  g_k(\bx-\by) \bJ(\by) dA_{\by} \, .
\label{Kdef}
\end{equation}

There is an enormous literature on the properties of these integral equations,
which we will not review here, except to note that the EFIE is poorly scaled;
one term in the representation of $\bE$ is of order $O(\omega)$ and one term is
of the order $O(\omega^{-1})$. This makes it difficult to compute both the
solenoidal and irrotational components of the current $\bJ$ and causes 
ill-conditioning in the integral equation at low frequencies --- a phenomenon
generally referred to as ``low-frequency breakdown" \cite{ZCCZ,Looptree4}.
Both the EFIE and the MFIE are also subject to spurious resonances at a 
countable set of frequencies $\omega_j$ going to infinity. Below the first
such resonance, the MFIE is a well-conditioned second kind Fedholm integral equation.
While low-frequency breakdown is obvious in the EFIE, it is not entirely 
avoided by switching to the MFIE \cite{ZCCZ}. 
The problem is that the current $\bJ$ is not sufficient for computing accurately the electric field. 
Note for example that 
\begin{equation}
\bn \cdot \bE = \rho = \frac{\nabla_{s} \cdot \bJ}{i \omega \epsilon} \, .
\label{escatnJ}
\end{equation}
As $\omega \rightarrow 0$, what in numerical analysis is called
{\em catastrophic cancellation} causes a progressive loss of digits
\cite{SPK1,SPK2}. 
Catastrophic cancellation comes not just from the ill-conditioning associated
with the evaluation of a derivative. The current $\bJ$ is an $O(1)$ quantity,
while $\nabla_s \cdot \bJ$ is $O(\omega)$, amplifying the loss of digits.
A variety of remedies to solve this problem have been suggested. In the widest
use are methods based on 
specialized basis functions for the discretization of the current $\bJ$ itself.
Loop-tree and loop-star basis functions, for example, can be used to rescale the 
solenoidal and the irrotational parts of the current 
\cite{Looptree5,Looptree3,Looptree1,Looptree2,Looptree4}.
A second class of methods is based on using both current and charge as separate
unknowns. This avoids terms of the order 
$O(\omega^{-1})$ (see \cite{CCIE1,CCIE3,CCIEour,CCIE2,CCIE4}). 
Unfortunately, all of these approaches encounter a second difficulty in multiply-connected
domains --- a phenomenon which we refer to as ``topological low-frequency breakdown'' 
\cite{MICH,CK1}.  At zero frequency, the MFIE, Calderon-preconditioned EFIE and
charge-current based integral equations are all rank-deficient, with a nullspace of dimension related to the topology of the surface $\partial D$: $g$ for the MFIE, $2g$ for the  Calderon-preconditioned EFIE and $g+N$ for the charge-current based integral equations \cite{CK1,MICH,WERNER}, where $g$ is the genus of the surface $\partial D$ and $N$ is the number of connected components.
This inevitably leads to ill-conditioning in the low-frequency regime.
This problem was carefully analyzed in the paper \cite{MICH}, and
the nullspace characterized in terms of harmonic vector fields
\cite{CK1,MICH,WERNER}. 

\begin{definition}
Assuming $D$ is topologically equivalent to a sphere with $g$ handles,
one can choose $g$ surfaces $S_j$ in $\bbR^3\backslash D$ so that 
$\bbR^3\backslash (D \cup_{j=1}^g S_j)$ is simply connected. The boundaries of
these surfaces are loops on $\partial D$ called {\em $B$-cycles}. They go 
around the ``holes" and form a basis for the first homology group of the 
domain $D$.  
\end{definition}

\begin{figure}[H]
\begin{center}
\includegraphics[width=5in]{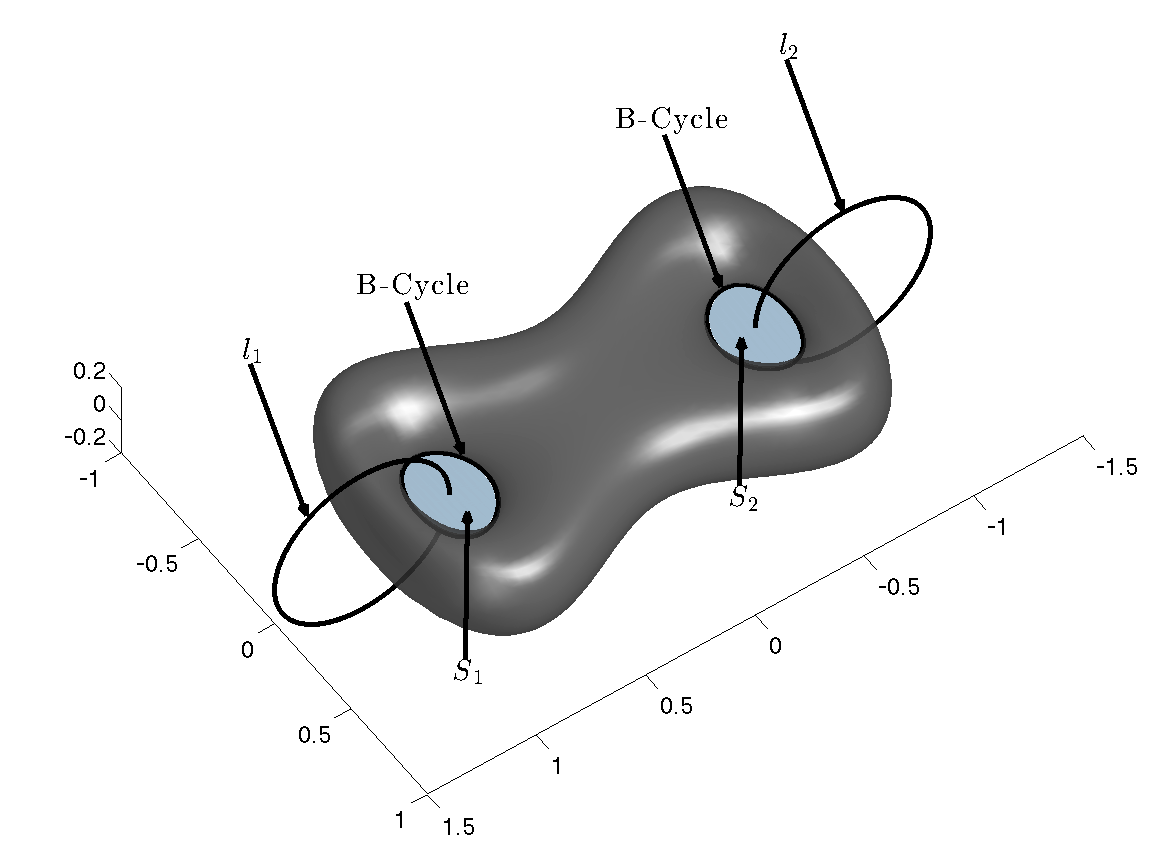}%
\end{center}
\caption{Double torus, number of connected components N=1, genus g=2, two B-cycles. $l_1$ and $l_2$ are the integration domain of integrals in (\ref{static_cond}). Surfaces $S_1,S_2$ are covering each hole of the surface. }
\label{Bcycles}
\end{figure}

In \cite{Mfietorus}, it was shown that the addition of 
a consistency condition of the form
\begin{equation}\label{stab_MFIE}
\begin{aligned}
 \int_{S_j}\bH^{\Sc} \cdot \bn \, da &=-\int_{S_j}\bH^{\In} \cdot \bn \, da
\end{aligned}
\end{equation}
or 
\begin{equation}
\int_{B_j} \bA^{\Sc} \cdot \ds = - \int_{B_j} \bA^{\In} \cdot \ds,
\label{Acond2}
\end{equation}
where the line integrals represent the circulation of the vector
potential, enforces uniqueness on the solution of the MFIE, assuming
$\omega$ is not at a spurious resonance.  This requires knowledge of,
or computation of, the genus, geometry tools that identify $B$-cycles,
and linear algebraic methods that are capable of efficiently solving
integral equations subject to constraints.

\begin{remark}
These topological issues can also be addressed through an analysis of the 
Hodge decomposition of the source current on the surface of the scatterer
itself, using the generalized Debye source representation of \cite{gDEBYE}.
Additional conditions are used (similar to (\ref{stab_MFIE})) to ensure
that the problem is well-posed.
\end{remark}

\begin{remark}
Very recently, a method was introduced 
that overcomes the topological low-frequency breakdown inherent in the EFIE by a
clever projection of 
the {\em discretized} problem using Rao-Wilton-Glisson (RWG) basis functions into a suitable subspace
\cite{MICHMULTI}. 
\end{remark}

In short, the various integral equations presently available pose
significant difficulties in the low-frequency regime.

\subsection{A decoupled formulation}
In this paper we introduce a new formulation for electromagnetic scattering
from perfect conductors. Rather than imposing boundary conditions on the field
quantities ($\bE$, $\bH$), we derive conditions on the potentials themselves.
Moreover, we show that the integral equations for $\bA^{\Sc}$
and $\phi^{\Sc}$ can be decoupled, lead to well-conditioned linear systems, and
are insensitive to the genus of the scatterer. More precisely, we seek to
impose the boundary conditions
\begin{equation}\label{Afitang}
\begin{array}{ll}
\bn\times\bA^{\Sc}(\bx)&=-\bn\times\bA^{\In}(\bx)|_{\partial D},\\
\bn\times\nabla \phi^{\Sc}(\bx)&=-\bn\times\nabla \phi^{\In}(\bx)|_{\partial D}.
\end{array}
\end{equation}

At first glance, there is an obvious difficulty with such an approach:
the vector and scalar potentials are not unique, a fact generally referred to as
{\em gauge freedom}. 
Even in the Lorenz gauge above,
the representation is known not to be unique. That is,
the condition (\ref{Lorenz}) does not completely determine the potentials 
$\bA^{\Sc},\phi^{\Sc}$.
To see this, consider the 
vector potentials $\bA'^{\Sc},\phi'^{\Sc}$ defined by
\begin{equation}\label{freedom}
\begin{aligned}
\bA'^{\Sc}[\bJ](\bx) &=\bA^{\Sc}[\bJ](\bx)+\nabla S_k[\sigma] (\bx),\\ 
 \phi'^{\Sc}[\rho](\bx) &= \phi^{\Sc}[\rho](\bx) + i\omega S_k[\sigma] (\bx).
 \end{aligned}
\end{equation}
Here, $\sigma$ is an arbitrary source on the surface $\partial D$.
It is straightforward to check that the fields 
$\bE^{\Sc},\bH^{\Sc}$ induced by $\bA'^{\Sc},\phi'^{\Sc}$ are the same as those
induced by $\bA^{\Sc},\phi^{\Sc}$, all the while satisfying the Lorenz gauge condition.

We will make use of this additional gauge freedom to establish a
well-posed boundary value problem and a stable, well-conditioned
integral equation.  In Section \ref{dpie:prel}, we consider the low-frequency 
limit of the exterior scattering from perfect conductors,
both for the sake of review and to motivate our formulation. In
Section \ref{dpie:sec}, we discuss the relevant existence, uniqueness
and stability results for what we refer to as the {\em decoupled
potential integral equation} (DPIE).  Finally, we discuss the stable
representation of the incoming field in terms of scalar and vector
potentials and the high-frequency behavior of the new
formulation.

\section{Preliminaries} \label{dpie:prel}

In this section, we consider the low-frequency limit of the Maxwell equations,
where the electric and magnetic fields are decoupled. We will refer to the 
electrostatic and magnetostatic fields by $\bE_0$ and $\bH_0$, respectively. 

\subsection{Electrostatics}\label{electrostatic}

The electrostatic field satisfies the equations
\begin{equation}
\label{estateqs}
\nabla\times\bE_0^{\tot}(\bx)=0, \quad
\nabla\cdot\bE_0^{\tot}(\bx)=0, \quad
\bx\in\mathbb{R}^3/D,
\end{equation}
which we decompose (as above) into incoming and scattered fields.
The scattered field must satisfy the radiation condition 
\begin{equation}\label{Estat}
\begin{array}{ll}
\bE_0^{\Sc}(\bx)=o(1), \quad |\bx|\rightarrow \infty.
\end{array}
\end{equation}
The boundary condition for the electrostatic field is
the same as that for any non-zero frequency,
\begin{equation}\label{etang0}
\bn\times\bE_0^{\tot}(\bx)=0, \quad \bx\in\partial D,
\end{equation}
but the solution is no longer unique.

The nullspace, that is functions satisfying
(\ref{estateqs}) and (\ref{etang0}) and the radiation condition (\ref{Estat}),
is of dimension $N$, where $N$ is the number of connected components of
the scatterers $\partial D$. 
They are known as {\em harmonic Dirichlet fields} \cite{CK1} with a basis
denoted by $\{\mathbf{Y}_j\}_{j=1}^N$.
It is straightforward to see that there are at least $N$ such solutions, since they 
correspond to the well-studied problem of capacitance.
To see this, let us denote by $\partial D_j$ the $j$th connected component of $\partial D$.
Taking the static limit of 
(\ref{Epotrep}), 
the scattered electrostatic field is described as the gradient of a scalar 
harmonic function:
\begin{equation}
\bE_0^{\Sc}=-\nabla \phi_0^{\Sc} \, .
\end{equation}
Imposing the boundary condition (\ref{etang0}) and assuming the incoming field
is represented in terms of an incoming potential $\phi^{\In}$, we have
\begin{eqnarray*}
\Delta \phi_0^{\Sc} &=& 0, \\
\bn\times\nabla\phi_0^{\Sc} &=& -\bn\times\nabla\phi_0^{\In}|_{\partial D}.
\end{eqnarray*}
It is clear that the preceding boundary condition is satisfied by any scattered
potential that satisfies the Dirichlet condition
\begin{equation}
\label{capdbc}
\phi_0^{\Sc}=-\phi_0^{\In}|_{\partial D_j}+V_j  \, ,
\end{equation}
where $V_j$ is an arbitrary constant on $\partial D_j$ that represents the voltage of each conductor (with respect to infinity).
The Dirichlet field $\mathbf{Y}_j$ corresponds to the 
gradient of $\phi_0^{\Sc}$ obtained by setting
$\phi_0^{\Sc}$ to zero on each boundary component $\partial D_i$ for $i \neq j$ and 
$\phi_0^{\Sc} = 1$ on $\partial D_j$.
Let us now define the scalars $Q_j$ by 
\begin{equation}
Q_j =  \int_{\partial D_j} \frac{\partial \phi_0^{\Sc}}{\partial n}ds=-\int_{\partial D_j} 
\mathbf{n}\cdot \mathbf{E}^{\Sc}_0 ds,
\label{capcharge}
\end{equation}
so that $-Q_j$ is the total charge on each conductor $\partial D_j$. The matrix that links the voltages $V_j$ and the charges $-Q_j$ is known as the {\em capacitance matrix} \cite{JACKSON}.
Since we are interested here in the time-harmonic Maxwell equations
and their zero-frequency limit, we must have charge neutrality on each
boundary component.  Thus, we are interested in studying
(\ref{capdbc}) where the voltages $V_j$ are additional unknowns, but for which
$N$ additional constraints are given of the form $Q_j=0$, for
$j=1,\dots,N$.

It is important to note that, in the static regime,
the problem suffers from more than non-uniqueness.
The boundary condition 
\[ \bn\times\nabla\phi_0^{\Sc}=-\bn\times\bE_0^{\In} \] cannot be
satisfied unless the incoming field is also an electrostatic field.
In particular, if the circulation of the incoming field
$\oint_{L\subset \partial D} \bE_0^{\In} \cdot d\mathbf{l}$ is not
zero on every closed loop $L$ on the surface $\partial D$, the
solution $\bE_0^{\Sc}$ does not exist. This follows 
easily from the fact that 
\[ \oint_{L\subset \partial D} \bE_0^{\In} \cdot d\mathbf{l}=
-\oint_{L\subset \partial D} \nabla\phi_0^{\Sc} \cdot d\mathbf{l}=0. \]
We refer the reader to \cite{CK1} for further discussion.

\subsection{Magnetostatics}

The magnetostatic field satisfies the equations
\begin{equation}
\label{hstateqs}
\nabla\times\bH_0^{\tot}(\bx)=0, \quad
\nabla\cdot\bH_0^{\tot}(\bx)=0, \quad
\bx\in\mathbb{R}^3/D,
\end{equation}
and the boundary condition 
\begin{equation}\label{hnorm0}
\bn\cdot\bH_0^{\tot}(\bx)=0|_{\partial D}.
\end{equation}
The total field is again decomposed into incoming and scattered fields, with
the scattered field satisfying the radiation condition (\ref{Hstat}),
\begin{equation}\label{Hstat}
\begin{array}{ll}
\bH_0^{\Sc}(\bx)=o(1), \quad |\bx|\rightarrow \infty.
\end{array}
\end{equation}
$\bH_0^{\Sc}$ can be described either as the curl of a harmonic vector potential 
$\mathbf{A_0}$ or as the gradient of a harmonic scalar potential $\phi_0^{\Sc}$.

The magnetostatic problem also suffers from non-uniqueness.
The nullspace, that is functions satisfying
(\ref{hstateqs}) and (\ref{hnorm0}) and the radiation condition (\ref{Hstat}),
is of dimension $g$, where $g$ is the genus of the surface 
$\partial D$. Elements of this space are called
{\em harmonic Neumann fields} $\{\mathbf{Z}_m\}_{m=1}^g$ \cite{CK1}. 
In order to completely specify the solution, additional
information, such as the total induced current $I_m$ on 
$g$ loops lying on the surface, must be specified:
\begin{equation}\label{static_cond}
\oint_{l_m}\bH_0^{\Sc}(\bx)\cdot d\mathbf{l}=I_m, \quad m=1,\ldots,g \, .
\end{equation}
The loops $l_m$ here go around the ``holes" (see Fig.
\ref{Bcycles}). That is, they are a basis for the first homology group
of $\bbR^3 \backslash D$, with spanning surfaces that lie in the {\em
  interior} of the scatterer, see \cite{CK1} for more details.  The
persistent currents $I_m$ at zero frequency are due to the potential presence of
superconducting loops (as we are considering scattering from perfect electric
conductors).

\subsection{Summary}

To summarize, the problem of 
electromagnetic scattering from perfect conductors is uniquely solvable for 
any $\omega$ strictly greater than zero.
At $\omega =0$, however, various subtleties arise.
The issue of Dirichlet fields needs to be resolved in electrostatics and 
the issue of Neumann fields needs to be resolved in magnetostatics.
For any $\omega$ strictly greater than zero, however, it is necessary that
the total charge $Q_j$ induced on any 
connected component of the scatterer be zero. Enforcing this condition at
$\omega=0$ 
(and introducing the additional unknown constants $V_j$ as above)
uniquely determines the electrostatic field.
In the magnetostatic case, however, we are obligated to introduce additional constants,
such as the $\{ I_m \}$ in (\ref{static_cond}), 
in order to account for the Neumann fields when the scatterer
has non-zero genus. 



\section{Scattering Theory for Decoupled Potentials} \label{dpie:sec}

We turn now to the analytic foundations of the DPIE. 
We first derive boundary value problems for the scattered
scalar and vector potentials that are completely insensitive to the genus, although they 
do depend explicitly on the number of boundary components.
After this reformulation of the Maxwell equations, we design integral
representations that lead to well-conditioned and invertible linear systems
of equations.

\begin{definition}
By the {\em scalar Dirichlet problem}, we mean the calculation of 
a scalar Helmholtz or Laplace potential in $\mathbb{R}^3\backslash D$
whose boundary value 
equals a given function $f$ on $\partial D$ and which satisfies 
standard radiation conditions at infinity:
\begin{equation}
\begin{aligned}
\Delta \phi^{\Sc} +k^2\phi^{\Sc}=0, \quad \phi^{\Sc}|_{\partial D}=f,\\
\end{aligned}
\end{equation}
\begin{equation}\label{radfiH}
\begin{array}{ll}
\frac{\bx}{|\bx|}\cdot\nabla \phi^{\Sc}(\bx)-ik\phi^{\Sc}(\bx)=o\Big(\frac{1}{|\bx|}\Big),\quad &|\bx|\rightarrow \infty\, ,
\end{array}
\end{equation}
for the scalar Helmholtz potential, and
\begin{equation}
\begin{aligned}
\Delta \phi_0^{\Sc} =0, \quad \phi_0^{\Sc}|_{\partial D}=f,\\
\end{aligned}
\end{equation}
\begin{equation}\label{radfiL}
\begin{array}{ll}
\phi^{\Sc}_0(\bx)=O\Big(\frac{1}{|\bx|}\Big),\quad 
\nabla \phi^{\Sc}_0(\bx)=O\Big(\frac{1}{|\bx|^2}\Big),\quad
&|\bx|\rightarrow \infty,
\end{array}
\end{equation}
for the scalar Laplace potential, respectively.
\end{definition}

\begin{definition}
By the {\em vector Dirichlet problem}, we mean the calculation of 
a vector Helmholtz or Laplace potential in $\mathbb{R}^3\backslash D$
whose tangential boundary
values equal a given tangential function $\mathbf{f}$ on $\partial D$, and whose
divergence equals a given scalar function $h$ on $\partial D$
and
which satisfies standard
radiation conditions at infinity:
\begin{equation}
\Delta \bA^{\Sc}+k^2\bA^{\Sc} =0, \quad
\bn\times \bA^{\Sc}|_{\partial D}=\mathbf{f}, \quad
\nabla \cdot \bA^{\Sc}|_{\partial D}=h, 
\end{equation}
\begin{equation}\label{radAH}
\begin{array}{ll}
\nabla\times \bA^{\Sc}(\bx)\times\frac{\bx}{|\bx|}+\frac{\bx}{|\bx|}\nabla\cdot \bA^{\Sc}(\bx)-ik\bA^{\Sc}(\bx)=o\Big(\frac{1}{|\bx|}\Big), \quad &|\bx|\rightarrow \infty\, ,
\end{array}
\end{equation}
for the vector Helmholtz potential, and 
\begin{equation}
\Delta \bA_0^{\Sc} =0, \quad
\bn\times \bA_0^{\Sc}|_{\partial D}=\mathbf{f}, \quad
\nabla \cdot \bA_0^{\Sc}|_{\partial D}=h, 
\end{equation}
\begin{equation}\label{radAL}
\begin{array}{ll}
\bA_0^{\Sc}(\bx)=O\Big(\frac{1}{|\bx|}\Big),\quad
\nabla\times \bA_0^{\Sc}(\bx)=O\Big(\frac{1}{|\bx|^2}\Big),\quad
\nabla\cdot \bA_0^{\Sc}(\bx)=O\Big(\frac{1}{|\bx|^2}\Big),
\quad &|\bx|\rightarrow \infty\, ,
\end{array}
\end{equation}
for the vector Laplace potential, respectively.
\end{definition}

For $k \neq 0$, both Dirichlet problems have unique solutions, but for
$k=0$, the vector Dirichlet problem has a nullspace --- the harmonic
Dirichlet fields discussed in Section \ref{electrostatic}. This lack
of uniqueness also makes the vector Dirichlet problem ill-conditioned 
at low frequencies.

\begin{table}[H] 
\centering
\begin{tabular}{c|cc}
\multicolumn{3}{c}{}\\
\hline
 Unknowns: & &  \\
 $\phi^{\Sc}$ & Laplace   & Helmholtz  \\
 $\bA^{\Sc}$ &   &   \\
\hline
\\
Scalar:  & $ \left\{ \begin{array}{r}
\Delta \phi^{\Sc} =0,\\
\phi^{\Sc}|_{\partial D}=f. 
\end{array} \right. $ (Yes) & $ \left\{ \begin{array}{r}
\Delta \phi^{\Sc} +k^2 \phi^{\Sc} =0,\\
\phi^{\Sc}|_{\partial D}=f. 
\end{array} \right. $ (Yes)\\
\\
Vector:  & $ \left\{ \begin{array}{r}
\Delta \bA^{\Sc} =0,\\
\bn\times \bA^{\Sc}|_{\partial D}=\mathbf{f},\\
\nabla \cdot \bA^{\Sc}|_{\partial D}=h.
\end{array} \right. $ (No) & $\left\{ \begin{array}{r}
\Delta \bA^{\Sc}+k^2\bA^{\Sc} =0,\\
\bn\times \bA^{\Sc}|_{\partial D}=\mathbf{f},\\
\nabla \cdot \bA^{\Sc}|_{\partial D}=h.
\end{array} \right. $ (Yes) \\
\\
\hline
\end{tabular}
\caption{Uniqueness for Dirichlet problems}
\label{tab1}
\end{table}

\subsection{Modified Dirichlet problems}

In order to address the non-uniqueness of the vector Dirichlet problem at
zero frequency and in order to enforce that the uncoupled scalar and vector
potentials define a suitable Maxwell field (enforcing the Lorenz gauge),
we introduce a related set of boundary value problems, 
which we refer to as {\em the modified Dirichlet problems}.  

\begin{definition}
By the {\em scalar modified Dirichlet problem}, we mean the calculation of 
a scalar Helmholtz or Laplace potential in $\mathbb{R}^3\backslash D$
which satisfies 
standard radiation conditions at infinity.
Letting $N$ denote the number of connected components of the boundary 
$\partial D$, we introduce extra unknown
degrees of freedom $\{V_j\}_{j=1}^N$
and the boundary data $f$ is supplemented with
additional (known) constants $\{Q_j\}_{j=1}^N$.
For the scalar Helmholtz potential, 
\begin{equation}
\begin{aligned}
\Delta \phi^{\Sc} +k^2\phi^{\Sc}=0, \quad \phi^{\Sc}|_{\partial D_j}=f+V_j,\\
\end{aligned}
\end{equation}
\begin{equation}\label{radfiHm}
\begin{array}{ll}
\frac{\bx}{|\bx|}\cdot\nabla \phi^{\Sc}(\bx)-ik\phi^{\Sc}(\bx)=o\Big(\frac{1}{|\bx|}\Big),\quad &|\bx|\rightarrow \infty\, ,
\end{array}
\end{equation}
with 
\[ \int_{\partial D_j} \frac{\partial \phi^{\Sc}}{\partial n}ds=Q_j. \]

For the scalar Laplace potential, 
\begin{equation}
\begin{aligned}
\Delta \phi_0^{\Sc} =0, \quad \phi_0^{\Sc}|_{\partial D_j}=f+V_j,\\
\end{aligned}
\end{equation}
\begin{equation}\label{radfiLm}
\begin{array}{ll}
\phi^{\Sc}_0(\bx)=O\Big(\frac{1}{|\bx|}\Big),\quad 
\nabla \phi^{\Sc}_0(\bx)=O\Big(\frac{1}{|\bx|^2}\Big),\quad
&|\bx|\rightarrow \infty,
\end{array}
\end{equation}
with 
\[ \int_{\partial D_j} \frac{\partial \phi^{\Sc}}{\partial n}ds=Q_j. \]
\end{definition}

\begin{definition}
By the {\em vector modified Dirichlet problem}, we mean the calculation of 
a vector Helmholtz or Laplace potential in $\mathbb{R}^3\backslash D$
which satisfies standard
radiation conditions at infinity.
Letting $N$ denote the number of connected components of the boundary 
$\partial D$, we introduce extra unknown
degrees of freedom $\{v_j\}_{j=1}^N$
and the boundary data $f$ is supplemented with
additional (known) constants $\{q_j\}_{j=1}^N$.
For the vector Helmholtz potential,
\begin{equation}
\Delta \bA^{\Sc}+k^2\bA^{\Sc} =0, \quad
\bn\times \bA^{\Sc}|_{\partial D}=\mathbf{f}, \quad
\nabla \cdot \bA^{\Sc}|_{\partial D_j}=h+v_j, 
\end{equation}
\begin{equation}\label{radAHm}
\begin{array}{ll}
\nabla\times \bA^{\Sc}(\bx)\times\frac{\bx}{|\bx|}+\frac{\bx}{|\bx|}\nabla\cdot \bA^{\Sc}(\bx)-ik\bA^{\Sc}(\bx)=o\Big(\frac{1}{|\bx|}\Big), \quad &|\bx|\rightarrow \infty\, ,
\end{array}
\end{equation}
with 
\[ \int_{\partial D_j} \bn\cdot \bA^{\Sc}ds=q_j. \]

For the vector Laplace potential,
\begin{equation}
\Delta \bA_0^{\Sc} =0, \quad
\bn\times \bA_0^{\Sc}|_{\partial D}=\mathbf{f}, \quad
\nabla \cdot \bA_0^{\Sc}|_{\partial D_j}=h+v_j, 
\end{equation}
\begin{equation}\label{radALm}
\begin{array}{ll}
\bA_0^{\Sc}(\bx)=O\Big(\frac{1}{|\bx|}\Big),\quad
\nabla\times \bA_0^{\Sc}(\bx)=O\Big(\frac{1}{|\bx|^2}\Big),\quad
\nabla\cdot \bA_0^{\Sc}(\bx)=O\Big(\frac{1}{|\bx|^2}\Big),
\quad &|\bx|\rightarrow \infty\, ,
\end{array}
\end{equation}
with 
\[ \int_{\partial D_j} \bn\cdot \bA^{\Sc}ds=q_j. \]
\end{definition}

We summarize the modified Dirichlet boundary value problems in Table \ref{tab2}.

\begin{table}[H] 
\centering
\begin{tabular}{c|cc}
\multicolumn{3}{c}{}\\
\hline
 Unknowns: & & \\
 $\phi^{\Sc}$, $\{V_j\}_{j=1}^N$&  Laplace  & Helmholtz  \\
 $\bA^{\Sc}$, $\{v_j\}_{j=1}^N$ &   &   \\
\hline
\\
Scalar:  & $\left\{ \begin{array}{r}
\Delta \phi^{\Sc} =0,\\
\phi^{\Sc}|_{\partial D_j}=f+V_j,\\
\int_{\partial D_j} \frac{\partial \phi^{\Sc}}{\partial n}ds=Q_j.
\end{array} \right.$ (Yes) & $\left\{ \begin{array}{r}
\Delta \phi^{\Sc} +k^2 \phi^{\Sc} =0,\\
\phi^{\Sc}|_{\partial D_j}=f +V_j,\\
\int_{\partial D_j} \frac{\partial \phi^{\Sc}}{\partial n}ds=Q_j.
\end{array} \right.$ (Yes) \\
\\
Vector:  & $\left\{ \begin{array}{r}
\Delta \bA^{\Sc} =0,\\
\bn\times \bA^{\Sc}|_{\partial D}=\mathbf{f},\\
\nabla \cdot \bA^{\Sc}|_{\partial D_j}=h+v_j,\\
\int_{\partial D_j} \bn\cdot \bA^{\Sc}ds=q_j.
\end{array} \right.$ (Yes) & $\left\{ \begin{array}{r}
\Delta \bA^{\Sc}+k^2\bA^{\Sc} =0,\\
\bn\times \bA^{\Sc}|_{\partial D}=\mathbf{f},\\
\nabla \cdot \bA^{\Sc}|_{\partial D_j}=h+v_j,\\
\int_{\partial D_j} \bn\cdot \bA^{\Sc}ds=q_j.
\end{array} \right.$ (Yes) \\
\\
\hline
\end{tabular}
\caption{Uniqueness for modified Dirichlet problems}
\label{tab2}
\end{table}

We now define the scattered scalar and vector potentials in terms of
modified Dirichlet problems.

\begin{definition} \label{defscatpot}
Let $\phi^{\In}, \bA^{\In}$ denote incoming scalar and vector 
potentials and assume that $D$ is a perfect conductor.
The {\em scattered scalar potential} $\phi^{\Sc}$ is the solution to the 
scalar modified Dirichlet problem with boundary data:
\begin{equation}
\begin{array}{cc}
f:=-\phi^{\In}|_{\partial D_j},& Q_j:=-\int_{\partial D_j}\frac{\phi^{\In}}{\partial n}ds 
\, .
\end{array}
\end{equation}
Likewise, the \emph{scattered vector potential} $\bA^{\Sc}$ is the solution to the 
vector modified Dirichlet problem with boundary data:
\begin{equation}
\begin{array}{ccc}
\mathbf{f}:=-\bn\times \bA^{\In}|_{\partial D},&h:=-\nabla\cdot \bA^{\In}|_{\partial D},& q_j:=-\int_{\partial D_j}\bn\cdot \bA^{\In}ds \, .
\end{array}
\end{equation}
\end{definition}

\subsection{Uniqueness}

We begin with two well-known theorems from scattering theory.

\begin{thm}\label{I_s}  \cite{CK1}
Let $\phi^{\Sc}$ be a scalar Helmholtz potential with wavenumber $k$, $(k\ne 0)$
in the exterior domain $\mathbb{R}^3\backslash \overline{D}$, 
satisfying the radiation condition (\ref{radfiH}) and the condition
\begin{equation}\label{I_s_eq}
I_s=\Im \bigg\{ k\int_{\partial D} \phi^{\Sc} \frac{\partial \overline{\phi}^{\Sc}}{\partial n}ds\bigg\}\ge 0.
\end{equation}
Here, $\Im \{f\}$ denotes the imaginary part of $f$.
Then, $\phi^{\Sc}=0$ in $\mathbb{R}^3\backslash \overline{D}$.
\end{thm}

\begin{thm}\label{I_v} \cite{CK1}
Let $\bA^{\Sc}$ be a vector Helmholtz potential with wavenumber $k$, $(k\ne 0)$ in the exterior 
domain $\mathbb{R}^3\backslash \overline{D}$,
satisfying the radiation condition (\ref{radAH}) and the condition
\begin{equation}\label{I_v_eq}
I_v=\Im \bigg\{ k\int_{\partial D} \bn\times \bA^{\Sc} \cdot \nabla \times \overline{\bA}^{\Sc}+\bn\cdot \bA^{\Sc}\nabla \cdot \overline{\bA}^{\Sc}ds\bigg\}\ge 0 \, .
\end{equation}
Then, $\bA^{\Sc}=0$ in $\mathbb{R}^3\backslash \overline{D}$.
\end{thm}

We now show that the modified Dirichlet problems have unique solutions
in all regimes.  For simplicity, we assume that $k$ is real. The proofs
are analogous when the wavenumber $k$ has a positive imaginary part,
which adds dissipation.

\begin{thm}\label{uniq_S}
The scalar modified Dirichlet problem has at most one solution for any $k>0$.
\end{thm}
\begin{proof}
Consider a solution of the homogeneous problem $(f=0,Q_j=0)$:
\begin{equation}
\left\{ \begin{array}{c}
\Delta \phi^{\Sc}+k^2\phi^{\Sc} =0,\\
\phi^{\Sc}|_{\partial D}=0+V_j,\\
\int_{\partial D_j} \frac{\partial \phi^{\Sc}}{\partial n}ds=0.
\end{array} \right.  
\end{equation}
The quantity $I_s$ in (\ref{I_s_eq}) is then given by
\begin{equation}
I_s=\Im \bigg\{ k\int_{\partial D} \phi^{\Sc} \frac{\partial \overline{\phi}^{\Sc}}{\partial n}ds\bigg\}=\Im \bigg\{ k\sum_{j=1}^N V_j \int_{\partial D_j}\frac{\partial \overline{\phi}^{\Sc}}{\partial n}\bigg\}=0.
\end{equation}
Thus, by Theorem \ref{I_s}, $\phi^{\Sc}=0$ and using the boundary condition $\phi^{\Sc}|_{\partial D_j}=0+V_j,$ we get $V_j=0$.
\end{proof}

\begin{thm}\label{uniq_V}
The vector modified Dirichlet problem has at most one solution for any $k>0$.
\end{thm}
\begin{proof}
Consider a solution of the homogeneous problem $(\mathbf{f}=0,h=0,q_j=0)$:
\begin{equation}
\left\{ \begin{array}{c}
\Delta \bA^{\Sc}+k^2\bA^{\Sc} =0,\\
\bn\times \bA^{\Sc}|_{\partial D}=\mathbf{0},\\
\nabla \cdot \bA^{\Sc}|_{\partial D_j}=0+v_j,\\
\int_{\partial D_j} \bn\cdot \bA^{\Sc}ds=0.
\end{array} \right. 
\end{equation}
The quantity $I_v$ in \ref{I_v_eq} is then given by
\begin{equation}
\begin{aligned}
I_v&=\Im \bigg\{ k\int_{\partial D} \bn\times \bA^{\Sc} \cdot \nabla \times \overline{\bA}^{\Sc}+\bn\cdot \bA^{\Sc}\nabla \cdot \overline{\bA}^{\Sc}ds\bigg\}=\\
&=\Im \bigg\{ k\sum_{j=1}^N \overline{v_j}\int_{\partial D_j} \bn\cdot \bA^{\Sc}ds\bigg\}=0.
\end{aligned}
\end{equation}
Thus, by Theorem \ref{I_v}, $\mathbf{A}^{\Sc}=0$ and using the boundary condition $\nabla \cdot \bA^{\Sc}|_{\partial D_j}=0+v_j$, we obtain $v_j=0$.
\end{proof}

\begin{thm}\label{uniq_SLS}
The scalar modified Dirichlet problem for the Laplace equation
has at most one solution.
\end{thm}
\begin{proof}
This is a well-known result. When $f=0$, the relation between $Q_j$ and $V_j$ 
is the capacitance matrix \cite{JACKSON}
and this matrix is always invertible (Theorem 5.6 in \cite{CK1}).
Thus, if $Q_j=0$, then $V_j=0$. The fact that $\phi^{\Sc}=0$ follows
from the maximum principle \cite{Evans}.
\end{proof}
Before proving the uniqueness of the vector modified Dirichlet problem for the Laplace equation we need the following technical result that shows a relation between the vector and scalar modified Dirichlet problems. This Lemma will also be used in section 4 to prove the connection between electromagnetic scattering and modified Dirichlet problems.

\begin{lemma}\label{th-divA}
Let $\bA^{\Sc},\{v_j\}_{j=1}^N$ be a solution of the vector modified Dirichlet
problem with boundary data $\mathbf{f},h,\{q_j\}_{j=1}^N$ for $(k\ge 0)$. Then,
\begin{equation}\label{conserv_L}
\begin{array}{cc}
\psi^{\Sc}:=\nabla\cdot \bA^{\Sc},& \{V_j=v_j\}_{j=1}^N
\end{array}
\end{equation}
satisfies the scalar modified Dirichlet problem with boundary data:
\begin{equation}
\begin{array}{cc}
f:=h,& \{Q_j=-k^2q_j\}_{j=1}^N \, .
\end{array}
\end{equation}
\end{lemma}

\begin{proof}
By hypothesis, $\bA^{\Sc}$ satisfies
\begin{equation}\label{helm}
\begin{aligned}
\Delta \bA^{\Sc}+k^2\bA^{\Sc} =0\\
\end{aligned}
\end{equation}
\begin{equation}
\begin{aligned}
\bn\times \bA^{\Sc}|_{\partial D}=\mathbf{f}\\
\end{aligned}
\end{equation}
\begin{equation}\label{divAbound}
\begin{aligned}
\nabla \cdot \bA^{\Sc}|_{\partial D_j}=h+v_j\\
\end{aligned}
\end{equation}
\begin{equation}\label{fluxth}
\begin{aligned}
\int_{\partial D_j} \bn\cdot \bA^{\Sc}ds=q_j \, .
\end{aligned}
\end{equation}
Taking the divergence of (\ref{helm}), we get
\begin{equation}
\begin{aligned}
\Delta \nabla\cdot \bA^{\Sc}+k^2\nabla\cdot \bA^{\Sc}=0\, .
\end{aligned}
\end{equation}
Therefore, $\psi^{\Sc}=\nabla\cdot \bA^{\Sc}$ satisfies the Helmholtz equation.
From \ref{divAbound}, we get
\begin{equation}
\begin{aligned}
\psi^{\Sc}=\nabla \cdot \bA^{\Sc}=h+v_j|_{\partial D_j}\, .
\end{aligned}
\end{equation}
Finally, we may write
\begin{equation}
\begin{aligned}
\nabla\times\nabla\times \bA^{\Sc}&=k^2\bA^{\Sc}+\nabla\nabla\cdot \bA^{\Sc}\Rightarrow\\
\bn\cdot\nabla\times\nabla\times \bA^{\Sc}&=k^2\bn\cdot\bA^{\Sc}+\bn\cdot\nabla\nabla\cdot \bA^{\Sc}\Rightarrow\\
-\nabla_s\cdot (\bn\times\nabla\times \bA^{\Sc})&=k^2\bn\cdot\bA^{\Sc}+\bn\cdot\nabla\frac{\partial \psi^{\Sc}}{\partial n}\Rightarrow\\
-\int_{\partial D_j}\nabla_s\cdot (\bn\times\nabla\times \bA^{\Sc}) ds=0&=k^2\int_{\partial D_j}\bn\cdot \bA^{\Sc} ds+\int_{\partial D_j}\frac{\partial\psi^{\Sc}}{\partial n} ds \, .
\end{aligned}
\end{equation}
Using the boundary condition (\ref{fluxth}), we obtain
\begin{equation}
\begin{aligned}
\int_{\partial D_j}\frac{\partial\psi^{\Sc}}{\partial n} ds=-k^2q_j\, ,
\end{aligned}
\end{equation}
and the result follows.
\end{proof}

\begin{thm}
The vector modified Dirichlet problem for the Laplace equation has at most one solution.
\end{thm}
\begin{proof}
Let $(\bA^{\Sc},v_j)$ be a solution of the homogeneous vector modified 
Dirichlet problem. 
Then, applying Lemma \ref{th-divA},  $(\nabla \cdot \bA^{\Sc},\{v_j\}_{j=1}^N)$ satisfies
the homogeneous scalar modified Dirichlet problem 
By Theorem \ref{uniq_SLS},  $\psi^{\Sc}=\nabla \cdot \bA^{\Sc}=0$ and $v_j=0$. 
By theorem 5.9 in \cite{CK1}, $\bA^{\Sc}$ is a harmonic Dirichlet field, and thus 
a linear combination of the basis functions $\{\mathbf{Y}_j\}_{j=1}^N$. 
It follows from the
flux conditions $\int_{\partial D_j} \bn\cdot \bA^{\Sc}ds=0$ that $\bA^{\Sc}=0$. 
\end{proof}

\subsection{Existence and stability}

In this section, we use the Fredholm alternative to obtain
existence results for the modified Dirichlet problems, making 
use of the single and double layer potentials, $S_k$ and $D_k$, 
of classical potential theory. We also show that the solution depends continuously on the boundary data,
uniformly in $k$ in a neighborhood of $k=0$.
Next we define classical operators in potential theory:

\begin{equation}\label{operators_scal}
\begin{aligned}
S_k\sigma=& \int_{\partial D} g_k(\bx-\by)\sigma(\by)dA_{\by}, \\
D_k\sigma=& \int_{\partial D} \frac{\partial g_k}{\partial n_y}(\bx-\by)\sigma(\by)dA_{\by}, \\
S'_k\sigma=&  \int_{\partial D} \frac{\partial g_k}{\partial n_x}(\bx-\by)\sigma(\by)dA_{\by}, \\
D'_k\sigma=&  \frac{\partial}{\partial n_x} \int_{\partial D} \frac{\partial g_k}{\partial n_y}(\bx-\by)\sigma(\by)dA_{\by},
\end{aligned}
\end{equation}
where $\bx\in\partial D$ and the Green's function on the free space is:
\begin{equation}\label{greens_function}
 g_k(\bx)=\frac{e^{ik|\bx|}}{4\pi|\bx|}.
\end{equation}
For off-surface evaluations $\bx\in\mathbb{R}^3\backslash\partial D$ we have:

\begin{equation}\label{operators_scal_off}
\begin{aligned}
S_k[\sigma](\bx)=& \int_{\partial D} g_k(\bx-\by)\sigma(\by)dA_{\by}, \\ 
D_k[\sigma](\bx)=& \int_{\partial D} \frac{\partial g_k}{\partial n_y}(\bx-\by)\sigma(\by)dA_{\by}.
\end{aligned}
\end{equation}

\begin{thm}\label{HSSth}
Suppose that we represent the solution to the scalar modified Dirichlet 
problem with $k>0$
in the form
\begin{equation}\label{fieldS1}
\phi^{\Sc}(\bx)=D_k[\sigma](\bx)-i\eta S_k[\sigma](\bx) \, ,
\end{equation}
with $\eta \in \bbR\backslash \{0\}$.
Then, 
imposing the desired boundary conditions and constraints leads to a Fredholm
equation of the second kind:
\begin{equation}\label{HSS1}
\begin{aligned}
\frac{\sigma}{2}+D_k\sigma-i\eta S_k\sigma-\sum_{j=1}^NV_j\chi_j=f,\\
\int_{\partial D_j} \big( D'_k\sigma+i\eta\frac{\sigma}{2} -i\eta S'_k\sigma \big) ds=Q_j \, ,
\end{aligned}
\end{equation}
where $\chi_j$ denotes the characteristic function for boundary
$\partial D_j$.
Here, $\sigma$ and the constants $\{V_j\}_{j=1}^N$ are unknowns. Moreover,
(\ref{HSS1}) is invertible and the result holds for the modified
Dirichlet problem governed by the Laplace equation ($k=0$) as well.
\end{thm}

\begin{proof}
See Appendix \ref{theoremsApp}.
\end{proof}

In order to study the vector modified Dirichlet problem, we define the
following dyadic operators:

\begin{equation}
\begin{aligned}
\overline{\overline{L}} 
\left(\begin{array}{c}\mathbf{a} \\ \rho \end{array} \right) 
=&\left( \begin{array}{c}
L_{11}\mathbf{a} + L_{12}\rho \\
L_{21}\mathbf{a} + L_{22}\rho
\end{array} \right),
\quad 
\end{aligned}
\end{equation}
where
\begin{equation}\label{operators_vect}
\begin{aligned}
L_{11}\mathbf{a}=&\ n \times S_k \mathbf{a}, \\
L_{12}\rho=&\ -\bn\times S_k(\bn\rho),\\
L_{21}\mathbf{a}=&\ 0,\\
L_{22}\rho=&\ D_k\rho,\\
\end{aligned}
\end{equation}
and
\begin{equation}
\begin{aligned}
\overline{\overline{R}} 
\left(\begin{array}{c}\mathbf{a} \\ \rho \end{array} \right) 
=&\left( \begin{array}{c}
R_{11}\mathbf{a} + R_{12}\rho \\
R_{21}\mathbf{a} + R_{22}\rho
\end{array} \right),
\end{aligned}
\end{equation}
where
\begin{equation}
\begin{aligned}
R_{11}\mathbf{a}=&\ \bn\times S_k(\bn\times\mathbf{a}),\\
R_{12}\rho=&\ \bn\times\nabla S_k(\rho),\\
R_{21}\mathbf{a}=&\ \nabla \cdot S_k(\bn\times\mathbf{a}),\\
R_{22}\rho=&\ -k^2S_k\rho.\\
\end{aligned}
\end{equation}

\begin{thm}\label{HVSth}
 Suppose that we represent the solution to the vector modified Dirichlet problem with $k>0$
in the form
\begin{equation}
\label{vecrep:DPIEv}
\bA^{\Sc}=\nabla\times S_k[\mathbf{a}](\bx)- S_k[\bn\rho](\bx)+i\eta \big( S_k[\mathbf{\bn\times a}](\bx)+ \nabla S_k[\rho](\bx) \big),
\end{equation}
with $\eta \in \bbR\backslash \{0\}$.
Then, for $|\eta|$ sufficiently small, 
imposing the desired boundary conditions and constraints leads to a Fredholm
equation of the second kind:
\begin{equation}\label{HVS1}
\begin{aligned}
\frac{1}{2}
\left(\begin{array}{c}\mathbf{a} \\ \rho \end{array} \right) 
+
\overline{\overline{L}} 
\left(\begin{array}{c}\mathbf{a} \\ \rho \end{array} \right) 
+
i\eta \overline{\overline{R}} 
\left(\begin{array}{c}\mathbf{a} \\ \rho \end{array} \right) 
+
\left(\begin{array}{c} 0 \\ \sum_{j=1}^{N} v_j \chi_j \end{array} \right) 
=
\left(\begin{array}{c} \mathbf{f} \\ h \end{array} \right), \\
\int_{\partial D_j} \Big( \bn\cdot\nabla \times S_k\mathbf{a}- 
\bn\cdot S_k(\bn\rho)+i\eta \big(\bn\cdot S_k(\mathbf{n\times a})-
\frac{\rho}{2}+S'_k \rho \big)  \Big) ds=q_j,
\end{aligned}
\end{equation}
where $\chi_j$ denotes the characteristic function for boundary $\partial D_j$.
Here, $\mathbf{a}, \rho$ and the constants $\{v_j\}_{j=1}^N$ are unknowns.
Moreover, (\ref{HVS1}) is invertible and the result holds for the vector modified Dirichlet
problem governed by the vector Laplace equation ($k=0$) as well.
\end{thm}

\begin{proof}
See Appendix \ref{theoremsApp}.
\end{proof}

\begin{definition} \label{dpiedef}
We will refer to (\ref{HSS1}) and (\ref{HVS1}) as the scalar and vector
decoupled potential integral equations. 
The former will be abbreviated by {\em DPIEs} and the latter by {\em DPIEv}.
Together, they form the {\em DPIE}.
\end{definition}

The following two theorems show that the solutions to the modified
Dirichlet problems are continuous functions of the boundary data all
the way to $k=0$.  In particular, they are independent of the genus of
$\partial D$.

\begin{thm}\label{unif_s}
The scalar modified Dirichlet problem has a unique solution for $k\ge0$. 
Moreover, the solution depends continuously on the boundary data  
$f,\{Q_j\}_{j=1}^N$ in the sense that the operator mapping the given boundary data 
onto the solution is uniformly continuous from 
\[ f,\{Q_j\}_{j=1}^N\in C^{0,\alpha}(\partial D)\times \mathbb{C}^N 
\quad \rightarrow \quad \phi^{\Sc},\{V_j\}_{j=1}^N\in C^{0,\alpha}(\mathbb{R}^3/D)\times \mathbb{C}^N 
\]
for any $k \in [0,k_{max}]$, with fixed $k_{max}$.
$C^{0,\alpha}(X)$ here is equipped with the usual H\"{o}lder norm \cite{CK1}.
\end{thm}

\begin{proof}
See Appendix \ref{theoremsApp}.
\end{proof}

\begin{thm}\label{unif_v}
The vector modified Dirichlet problem has a unique solution for $k\ge0$. 
Moreover, the solution depends continuously on the boundary data  
$\mathbf{f},h,\{q_j\}_{j=1}^N$ in the sense that the operator mapping the given 
boundary data onto the solution is uniformly continuous from 
\[ \mathbf{f},h,\{Q_j\}_{j=1}^N\in T^{0,\alpha}(\partial D)\times 
C^{0,\alpha}(\partial D)\times \mathbb{C}^N \quad\rightarrow\quad 
\bA^{\Sc},\{v_j\}_{j=1}^N\in C^{0,\alpha}(\mathbb{R}^3/D)\times \mathbb{C}^N
\]
for any $k \in [0,k_{max}]$, with fixed $k_{max}$.
Here, $T^{0,\alpha}(\partial D)$ is equipped with the usual Holder norm \cite{CK1}.
\end{thm}

\begin{proof}
See Appendix \ref{theoremsApp}.
\end{proof}


\section{Electromagnetic scattering and 
modified Dirichlet problems}

In this section, we explain the connection between the scalar and
vector modified Dirichlet problems and the Maxwell equations.  It is
evident from Theorems \ref{unif_s} and \ref{unif_v} that, if such a
reformulation exists, then we have overcome the topological
low-frequency breakdown that makes electromagnetic scattering from
surfaces with nontrivial genus so difficult at low
frequency. 

We will first show that the vector
and scalar modified Dirichlet problems preserve the Lorenz gauge, so 
that the induced $\bE$ and $\bH$ fields are Maxwellian. We will also show
that the calculation is stable, in the sense that bounded ``incoming" data leads to 
bounded ``outgoing" data, independent of the frequency.
We will then show, in Theorem \ref{mainthm}, that the modified Dirichlet
problems lead directly to the solution of the desired scattering problem.

\begin{thm}\label{conserv} 
Let $\bA^{\In},\phi^{\In}$ be bounded (for $\omega\rightarrow 0$) incoming vector and scalar Helmholtz
potentials in the Lorenz gauge:
\begin{equation}
\mathbf{\nabla\cdot A^{\In}}=i\omega\mu\epsilon\phi^{\In} \, .
\end{equation}
Then, the associated vector and scalar scattered Helmholtz potentials
$\bA^{\Sc},\phi^{\Sc}$ (see Definition \ref{defscatpot}) are also bounded and satisfy the Lorenz
gauge condition.
\end{thm}

\begin{proof}
By Lemma \ref{th-divA}, the scalar Helmholtz potential
$\psi^{\Sc}=\nabla\cdot \bA^{\Sc}$ satisfies
\begin{equation}
\begin{aligned}
\psi^{\Sc}&=h+v_j=-\nabla\cdot \bA^{\In}+v_j\big|_{\partial D}\\
\int_{\partial D_j} \frac{\partial \psi^{\Sc}}{\partial n}ds &=-k^2q_j=k^2\int_{\partial D_j} \mathbf{n\cdot A^{\In}}ds \, .
\end{aligned}
\end{equation}
Using the Lorenz gauge condition on the boundary itself, we may write
\begin{equation}\label{firstcond}
\psi^{\Sc}=h+v_j=-i\omega\mu\epsilon\phi^{\In}+v_j\big|_{\partial D} \, .
\end{equation}
Since
\begin{equation}
\begin{aligned}
\nabla \cdot \big( i\omega\bA^{\In}-\nabla\phi^{\In}\big)=0 \, ,
\end{aligned}
\end{equation}
we have
\begin{equation}
\int_{\partial D_j} \frac{\partial \phi^{\In}}{\partial n}ds =i\omega\int_{\partial D_j} \mathbf{n\cdot A^{\In}}ds \, .
\end{equation}
Thus,
\begin{equation}\label{secondcond}
\int_{\partial D_j} \frac{\partial \psi^{\Sc}}{\partial n}ds =k^2\int_{\partial D_j} \mathbf{n\cdot A^{\In}}ds=\frac{k^2}{i\omega}\int_{\partial D_j} \frac{\partial \phi^{\In}}{\partial n}ds=-i\omega\mu\epsilon\int_{\partial D_j} \frac{\partial \phi^{\In}}{\partial n}ds \, .
\end{equation}
From (\ref{firstcond}) and (\ref{secondcond}),
we see that $\psi^{\Sc}$ and $i\omega\mu\epsilon\phi^{\Sc}$ satisfy
the same scalar modified Dirichlet problem.
By uniqueness (Theorem \ref{uniq_S}), we find that 
\[ i\omega\mu\epsilon\phi^{\Sc}=\psi^{\Sc}=\nabla\cdot \bA^{\Sc} \, , \]
so that $\bA^{\Sc}$ and $\phi^{\Sc}$ are in the Lorenz gauge.
By Theorems \ref{unif_s} and \ref{unif_v}, $\bA^{\Sc},\phi^{\Sc}$ are uniformly 
continuous functions 
of $\bA^{\In},\phi^{\In}$ for $k\in[0,k_{max}]$. Since
$\bA^{\In},\phi^{\In}$ are bounded, $\bA^{\Sc},\phi^{\Sc}$ are also bounded.
\end{proof}

The next theorem is the main result of the present paper. 

\begin{thm} \label{mainthm}
For any $k\ge0$, let $\mathbf{E^{\In},H^{\In}}$ be an incoming electromagnetic field 
described by the potentials $\bA^{\In},\phi^{\In}$ in the Lorenz gauge:
\[
\begin{aligned}
\bH^{\In}&=\frac{1}{\mu}\nabla\times \bA^{\In},\\
\bE^{\In}&=i\omega\bA^{\In}-\nabla \phi^{\In}, \\
\nabla\cdot\bA^{\In} &=i\omega\mu\epsilon \phi^{\In} \, ,
\end{aligned}
\]
and let 
$\bA^{\Sc},\phi^{\Sc}$ denote the corresponding scattered vector and scalar
potentials (Definition  \ref{defscatpot}).
Then the electromagnetic fields 
$\mathbf{E^{\Sc},H^{\Sc}}$ scattered from a perfect conductor are given by
\[
\begin{aligned}
\bH^{\Sc}&=\frac{1}{\mu}\nabla\times \bA^{\Sc},\\
\bE^{\Sc}&=i\omega\bA^{\Sc}-\nabla \phi^{\Sc}.
\end{aligned}
\]
with
\[ \nabla\cdot\bA^{\Sc} =i\omega\mu\epsilon \phi^{\Sc}, \]
\[
\begin{aligned}
\bn\times \bE^{\Sc}&=-\bn\times \bE^{\In}|_{\partial D}, \quad
\bn\cdot \bH^{\Sc}&=-\bn\cdot \bH^{\In}|_{\partial D}.
\end{aligned}
\]
\end{thm}

\begin{proof}
Since $\bA^{\Sc},\phi^{\Sc}$ are Helmholtz potentials in the Lorenz gauge, 
the associated $\bE^{\Sc},\bH^{\Sc}$ are valid Maxwell fields that satisfy the 
necessary radiation condition. We need only check that the desired boundary conditions
are satisfied.
From the boundary conditions on $\bA^{\Sc},\phi^{\Sc}$ we have
\begin{equation}\label{potA1}
\begin{aligned}
\bn\times \bA^{\Sc}&=-\bn\times \bA^{\In}|_{\partial D}\\
\Rightarrow i\omega\bn\times \bA^{\Sc}&=-i\omega\bn\times \bA^{\In}|_{\partial D} \, ,
\end{aligned}
\end{equation}

\begin{equation}\label{potphi1}
\begin{aligned}
\phi^{\Sc}&=-\phi^{\In}+V_j|_{\partial D_j}\\
\Rightarrow \bn\times\nabla \phi^{\Sc}&=-\bn\times\nabla \phi^{\In}|_{\partial D} \, .
\end{aligned}
\end{equation}
Adding (\ref{potA1}) and (\ref{potphi1}), we have 
\begin{equation}\label{potA}
\begin{aligned}
i\omega\bn\times \bA^{\Sc}-\bn\times\nabla \phi^{\Sc}&=-i\omega\bn\times \bA^{\In}+\bn\times\nabla \phi^{\In}|_{\partial D}\\
\Rightarrow \bn\times \bE^{\Sc}&=-\bn\times \bE^{\In}|_{\partial D} \, .
\end{aligned}
\end{equation}
Taking the surface divergence of (\ref{potA1}), we also have that
\begin{equation}
\begin{aligned}
\bn\times \bA^{\Sc}&=-\bn\times \bA^{\In}|_{\partial D}\\
\Rightarrow \nabla_s \cdot\bn\times \bA^{\Sc}&=-\nabla_s \cdot\bn\times \bA^{\In}|_{\partial D}\\
\Rightarrow \bn\cdot \bH^{\Sc}&=-\bn\cdot \bH^{\In}|_{\partial D} \, .
\end{aligned}
\end{equation}

Thus, for $k>0$ we have the correct solution. While continuity arguments are sufficient
to verify that the zero frequency solution is the desired one, it is worth checking that
the net charge still vanishes at $k=0$ and that the consistency conditions
(\ref{stab_MFIE}) are satisfied. For this, note first that
$\mathbf{A},\phi$ are bounded, so that at $k=0$,
\begin{equation}
\begin{aligned}
\bH_0^{\In}&=\lim_{k\rightarrow 0}\bH^{\In}= \nabla\times \bA^{\In}_0,\\
\bE_0^{\In}&=\lim_{k\rightarrow 0}\bE^{\In}=-\nabla \phi^{\In}_0,\\
\bH_0^{\Sc}&=\lim_{k\rightarrow 0}\bH^{\Sc}=\nabla\times \bA^{\Sc}_0,\\
\bE_0^{\Sc}&=\lim_{k\rightarrow 0}\bE^{\Sc}=-\nabla \phi^{\Sc}_0 \, .
\end{aligned}
\end{equation}
The net charge is computed as the surface integral of $\bn \cdot \bE$, and we have
\begin{equation}
\begin{aligned}
\int_{\partial D_j}\bn\cdot \bE_0^{\Sc}ds &=-\int_{\partial D_j}\frac{\partial \phi^{\Sc}_0}{\partial n}ds=\int_{\partial D_j}\frac{\partial \phi^{\In}_0}{\partial n}ds=\int_{D_j}\Delta \phi^{\In}_0dv=0,\\
\int_{\partial D_j}\bn\cdot \bE^{\Sc}ds &=\int_{\partial D_j}i\omega\bn\cdot \bA^{\Sc}-\frac{\partial \phi^{\Sc}}{\partial n}ds=\\
&=-\int_{\partial D_j}i\omega\bn\cdot \bA^{\In}-\frac{\partial \phi^{\In}}{\partial n}ds=\int_{D_j}\nabla \cdot \big( i\omega\bA^{\In}-\nabla\phi^{\In}\big)dv=0 \, .
\end{aligned}
\end{equation}
The last equality follows from the fact that the
incoming potentials are assumed to be specified in the Lorenz
gauge. In short,
\begin{equation}
\begin{aligned}
&\int_{\partial D_j}\bn\cdot \bE_0^{\Sc}ds=\lim_{k\rightarrow 0}\int_{\partial D_j}\bn\cdot \bE^{\Sc}ds=0 \, ,
\end{aligned}
\end{equation}
as expected.
The consistency conditions (\ref{stab_MFIE}) on the flux of 
the magnetic field through each hole 
$S_j$ \cite{Mfietorus} are also easily verified for all $\omega\ge0$:
\begin{equation}
\begin{aligned}
\bn\times \bA^{\Sc}&=-\bn\times \bA^{\In}|_{\partial D} \\ 
\Rightarrow \oint_{B_j}\bA^{\Sc} \cdot d\mathbf{l}
&=-\oint_{B_j}\bA^{\In} \cdot d\mathbf{l}\, .
\end{aligned}
\end{equation}
Here, $B_j$ is a  $B$-cycle, namely a 
loop on $\partial D$ which goes around some ``hole" and whose
spanning surface $S_j$ 
lies in the {\em exterior} of the domain (see Figure \ref{Bcycles} and related discussion).
\end{proof}

\section{Incoming Potentials} \label{incoming:sec}

In a stable DPIE approach, the vector and scalar potentials must be
defined in the Lorenz gauge and be bounded as $\omega \rightarrow 0$.
We will need to find a representation for the incoming fields that
will permit the stable uncoupling of the vector and scalar potentials.
Assuming we are given the ``impressed" free current and charge
$\bJ^{imp},\rho^{imp}$, the incoming potentials
\begin{equation} \label{impressed}
\begin{aligned}
\bA^{\In}(\bx)&=\mu S_k[\bJ^{imp}](\bx),\\
\phi^{\In}(\bx)&=\frac{1}{\epsilon}S_k[\rho^{imp}](\bx),
\end{aligned}
\end{equation}
satisfy these requirements. For an incoming plane wave with a polarization vector $\bE_p$ and a direction of
propagation $\mathbf{u}$, given by
\begin{equation} \label{emplanew}
\begin{aligned}
\bE^{\In}=\bE_pe^{ik\mathbf{u\cdot x}},\quad \bH^{\In}=\bH_pe^{ik\mathbf{u\cdot x}}=\frac{\mathbf{u}\times \bE_p}{Z} e^{ik\mathbf{u\cdot x}} ,
\end{aligned}
\end{equation}
where $Z=\sqrt{\frac{\mu}{\epsilon}}$, the standard representation of incoming vector and scalar potentials
\begin{equation}
\bA^{\In}=\frac{1}{i\omega} \bE^{\In}, \quad \phi^{\In}=0,
\end{equation}
does not lead to stable uncoupling, since the vector potential is unbounded, as $\omega \to 0$. But, as mentioned above, the Lorenz gauge
does not, by itself, impose uniqueness on the governing potentials.
It is easy to check that the vector and scalar potentials defined by
\begin{equation}
\begin{aligned}
\bA'^{\In}&=-\mathbf{u}(\bx\cdot \bE_p)\sqrt{\mu\epsilon}e^{ik\mathbf{u}\cdot \bx},\\
\phi'^{\In}&=-\mathbf{\bx}\cdot \bE_pe^{ik\mathbf{u}\cdot \bx},
\end{aligned}
\end{equation}
satisfy the Lorenz gauge condition
\begin{equation}\label{planewaveEM}
\begin{aligned}
\nabla\cdot\bA'^{\In}(\bx)=i\omega\mu\epsilon\phi'^{\In}(\bx),
\end{aligned}
\end{equation}
both $\bA'^{\In}$ and $\phi'^{\In}$ are bounded Helmholtz potentials, as $\omega \to 0$,
and represent the same incoming plane wave (\ref{emplanew}). See
Appendix \ref{Partial_wave_expansion} for more details how to stably
decompose incoming/outgoing electric and magnetic multipole fields
(Debye sources).

\section{The DPIE and the Aharonov-Bohm effect}

In classical physics, the Maxwell equations are described in terms of the
components of the electric and magnetic fields, with the vector and
scalar potentials viewed as matters of computational convenience.
In quantum mechanics, however, it was shown by Aharonov and Bohm 
\cite{Aharonovbohm} that an electron is sensitive to the vector potential $\bA$
itself, in regions where $\bE$ and $\bH$ are identically zero (the Aharonov-Bohm effect).

Let us first recall that the two pairs of potentials
$\{\bA,\phi\}$ and $\{\bA',\phi'\}$ produce the same electromagnetic field, so long as
they satisfy the condition
\begin{equation}\label{gaugecond2}
\begin{aligned}
\mathbf{A'}&=\mathbf{A}+\nabla \psi,\\
\phi'&=\phi+i\omega\psi \, .
\end{aligned}
\end{equation}

In multiply connected regions at zero frequency, however, the situation is more complex.
There exist potentials which give rise to identical fields that are not related
according to (\ref{gaugecond2}). In particular, the potentials 
\begin{equation}\label{zerononzero}
\begin{aligned}
\bA_0&=\mathbf{Z_1},&\phi_0=0, \\
\bA'_0&=\mathbf{0},&\phi'_0=0,
\end{aligned}
\end{equation}
where $\mathbf{Z_1}$ is an exterior harmonic Neumann field, give rise
to zero electromagnetic fields in the exterior.  $\mathbf{Z_1}$,
however, is not the gradient of a single-valued harmonic function.

\begin{figure}[H]
\begin{center}
\includegraphics[width=3in]{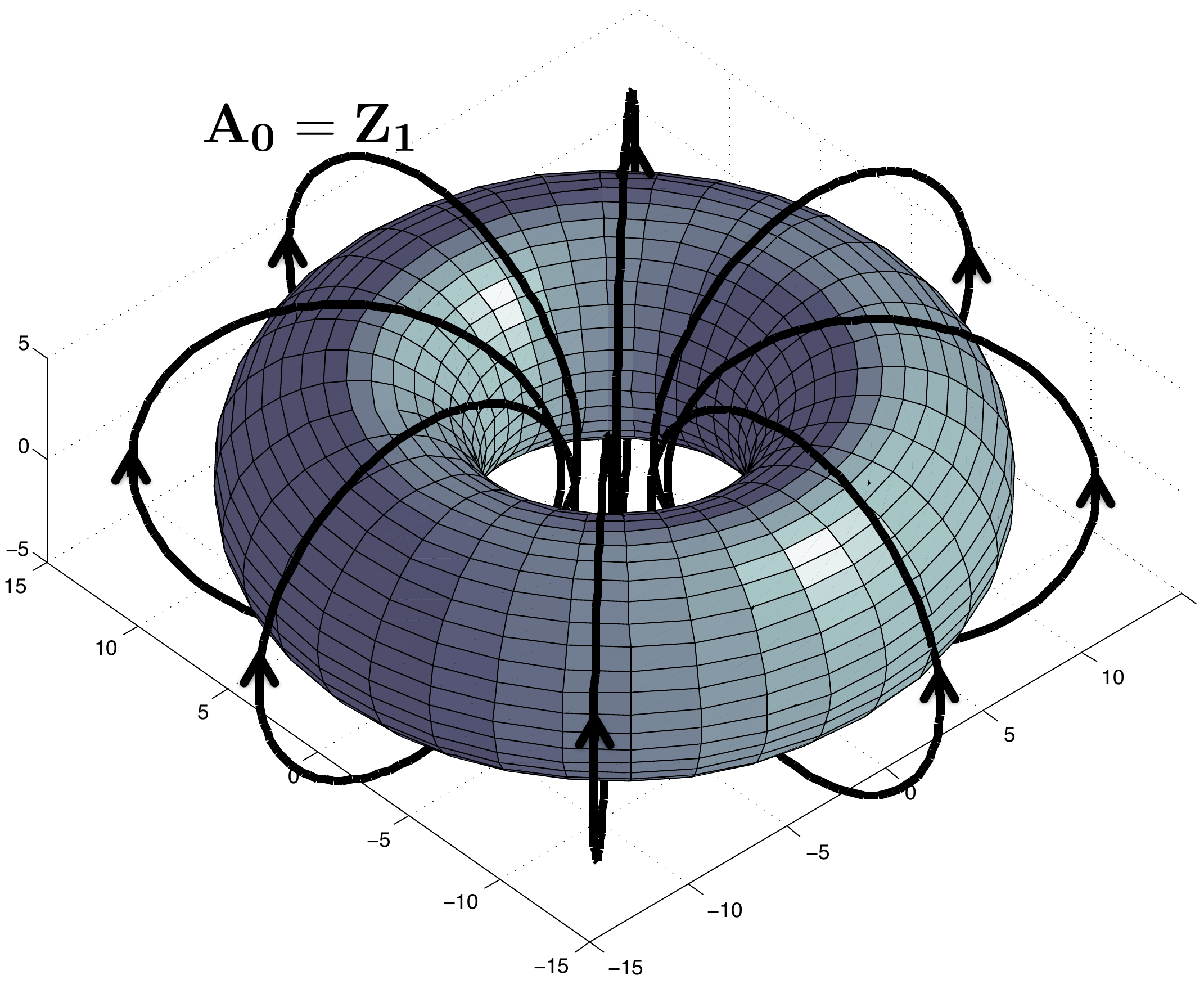}%
\end{center}
\caption{In the exterior of a torus, a harmonic Neumann field $Z_1$ serves as a
vector potential $\bA_0$ with the corresponding scalar potential set equal to zero.
For $\bx\in\mathbb{R}^3\backslash D$, the associated electromagnetic
fields $\bE$ and $\bH$ are identically zero.}
\label{fig_Neumann}
\end{figure}

The Aharonov-Bohm effect is based on an experiment that is able to 
distinguish between the physical states  $\bA_0,\phi_0$ and  $\bA'_0,\phi'_0$.
We have taken some liberties with the actual experiment
in \cite{Aharonovbohm} but the physical idea is the same.
In essence, quantum mechanical tunneling permits an electron to be aware of the 
electromagnetic field in the {\em interior} of $D$, even though it is a perfect
conductor. For $\bA'_0,\phi'_0$, the 
field is identically zero in the interior, but for 
$\bA_0,\phi_0$, it is not. As discussed in \cite{CK1,MICH,gDEBYE,WERNER},
$\bA_0$ can be viewed as the field induced by an axisymmetric current density 
flowing on the surface in the direction of the arrows in Fig. \ref{fig_Neumann}. 
This induces a non-trivial magnetic field within the torus. Electrons, as a result,
sense whether they traveled through the hole of the torus or passed by the torus
on the outside. 
The DPIE formalism easily distinguishes between these two cases, since
we deal with the vector and scalar potentials directly.
Thus, $\bA_0,\phi_0,\bA'_0$ and $\phi'_0$ in (\ref{zerononzero}) satisfy
\begin{equation}
\begin{aligned}
\bn\times \bA_0|_{\partial D}&=\bn\times \mathbf{Z_1},&\nabla\cdot \bA_0|_{\partial D}&=0,&\quad \phi_0|_{\partial D}=0,\\
\bn\times \bA'_0|_{\partial D}&=\mathbf{0},&\nabla\cdot \bA'_0|_{\partial D}&=0,& \quad \phi'_0|_{\partial D}=0\, .
\end{aligned}
\end{equation}

\section{The DPIE in the high frequency regime}

Theorems \ref{HSSth} and \ref{HVSth} suggest that the numerical
solution of scattering problems in the presence of perfect conductors
can be effectively solved through the use of the DPIE, defined by
equations (\ref{HSS1}) and (\ref{HVS1}).  The scalar part of these
equations (\ref{HSS1}) is a Fredholm equation of the second kind,
invertible for all frequencies.  As stated in Theorem \ref{HVSth},
however, the vector part (\ref{HVS1}) is a Fredholm equation only for
sufficiently small coupling constant
$|\eta|<\|\overline{\overline{R}}\|^{-1}$.  The difficulty is that the
operator $\overline{\overline{R}}$ is continuous and bounded, but not
compact.  In fact, its spectrum has three cluster points:
$\lambda=0.5, \lambda=0.5+i 0.5$ and $\lambda=0.5-i 0.5$ (see
\cite{Bruno_Krylov,Nedelec}).  While the uniqueness proof holds for
arbitrary values of $\eta$, existence requires further analysis. We
suspect that the vector part of DPIE is invertible for all frequencies
and leave the formal result as a conjecture.  In fact, numerical
experiments suggest that $\eta$ should not be chosen too small when
seeking to optimize the condition number of the DPIE.

Our interest in the DPIE formulation grew out of issues in low-frequency scattering.
Nevertheless, we would like to find a representation that is effective at
all frequencies, and this will involve a slight rescaling of the equations. 
In order to carry out a suitable analysis, we follow \cite{Kress_sphere} and
study scattering from the unit sphere $\partial D=\{\bx : \|\bx\|=1\}$. 
For $k \leq 1$, setting $\eta = 1$ works well, while for $k>1$ the optimal
scaling factor $\eta \approx k$ (see \cite{Kress_sphere}).
Setting $\eta = k$, instead of (\ref{HSS1}), we have the 
{\em scaled} DPIEs integral equation:
\begin{equation}\label{HSS}
\begin{aligned}
\frac{\sigma}{2}+D_k\sigma-i k S_k\sigma-\sum_{j=1}^NV_j\chi_j=f,\\
\int_{\partial D_j} \big( \frac{1}{k} D'_k\sigma+i\frac{\sigma}{2} -i S'_k\sigma \big) ds=\frac{1}{k} Q_j \, ,
\end{aligned}
\end{equation}
where the second set of equations has been multiplied by a factor of $\frac{1}{k}$. 
For the vector modified Dirichlet problem, when $k>1$, we replace
(\ref{vecrep:DPIEv}) with
\begin{equation} 
\bA^{\Sc}=\nabla\times S_k[\mathbf{a}](\bx)- kS_k[\bn\varrho](\bx)+
i \big( k S_k[\mathbf{\bn\times a}](\bx)+ \nabla S_k[\varrho](\bx) \big),
\end{equation}
where we have multiplied the single-layer potential terms by $k$, and 
set $\eta = 1$.
We also rescale the boundary condition 
$\nabla \cdot \bA^{\Sc} = -\nabla \cdot \bA^{\Sc} + v_n$ in the modified Dirichlet
problem, dividing each side by $k$.
These changes lead to the {\em scaled DPIEv} integral equation:
\begin{equation}\label{DPIEv1}
\begin{aligned}
\frac{1}{2}
\left(\begin{array}{c}\mathbf{a} \\ \rho \end{array} \right) 
+
\overline{\overline{L_s}} 
\left(\begin{array}{c}\mathbf{a} \\ \rho \end{array} \right) 
+
i \overline{\overline{R_s}} 
\left(\begin{array}{c}\mathbf{a} \\ \rho \end{array} \right) 
+
\left(\begin{array}{c} 0 \\ \sum_{j=1}^{N} v_j \chi_j \end{array} \right) 
=
\left(\begin{array}{c} \mathbf{f} \\ \frac{1}{k} h \end{array} \right), \\
\int_{\partial D_j} \Big( \bn\cdot\nabla \times S_k\mathbf{a}-k\bn\cdot S_k(\bn\varrho)+i \big(k\bn\cdot S_k(\mathbf{n\times a})-\frac{\varrho}{2}+S'_k \varrho \big)  \Big) ds=q_j,\\
\end{aligned}
\end{equation}
where
\begin{equation}
{\overline{\overline{L_s}}}=\left( \begin{array}{cc}
L_{11} & kL_{12}\\
\frac{1}{k}L_{21} & L_{22}
\end{array} \right),
\ \
{\overline{\overline{R_s}}}=\left( \begin{array}{cc}
kR_{11} & R_{12}\\
R_{21} & \frac{1}{k}R_{22}
\end{array} \right),
\end{equation}
with 
$L_{ij}, R_{ij}$ defined in (\ref{operators_vect}).

We turn now to the analysis of the DPIE on the unit sphere, where
exact expressions for the various integral operators 
have been worked out in detail \cite{SPHANAL}. More precisely, 
using scalar and vector spherical harmonics, each
integral operator has a simple signature, which has been tabulated in \cite{SPHANAL}.
This permits us to compute the condition number and spectrum of the DPIEs and DPIEv
integral equations.
In Figs. \ref{fig_Eig} and \ref{fig_svd}, we plot the spectrum and the singular 
values of the {\em scaled} DPIEv (\ref{DPIEv1}).

\begin{figure}[H]
\begin{center}
\includegraphics[width=3in]{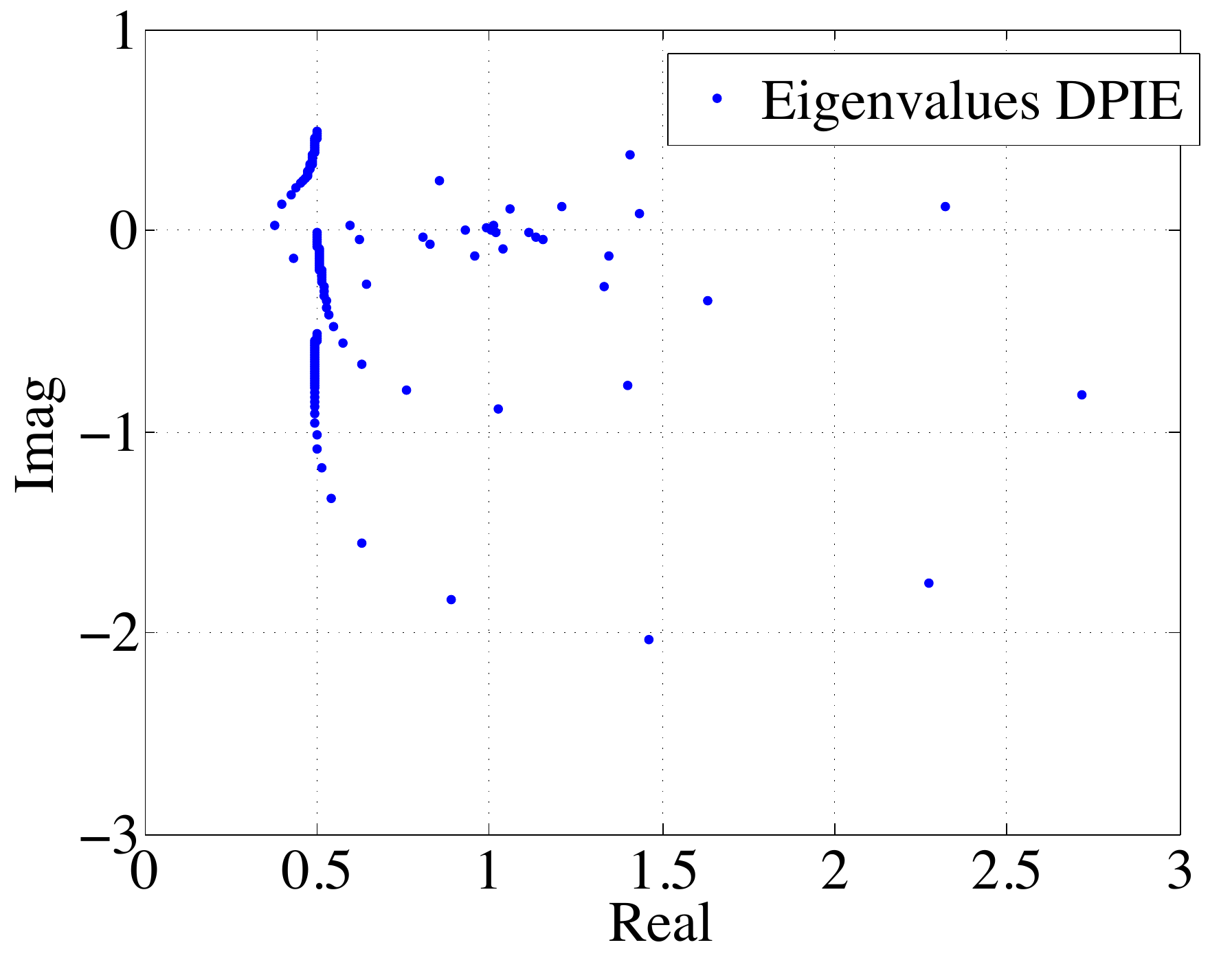}%
\end{center}
\caption{Spectrum of the scaled DPIEv integral equation (\ref{DPIEv1}) for 
a spherical scatterer of radius 1 at $k=10$.
As discussed in the text, there are three different cluster points:
at $\lambda=0.5, \lambda=0.5+i0.5$ and $\lambda=0.5-i0.5$.}
\label{fig_Eig}
\end{figure}

\begin{figure}[H]
\begin{center}
\includegraphics[width=3in]{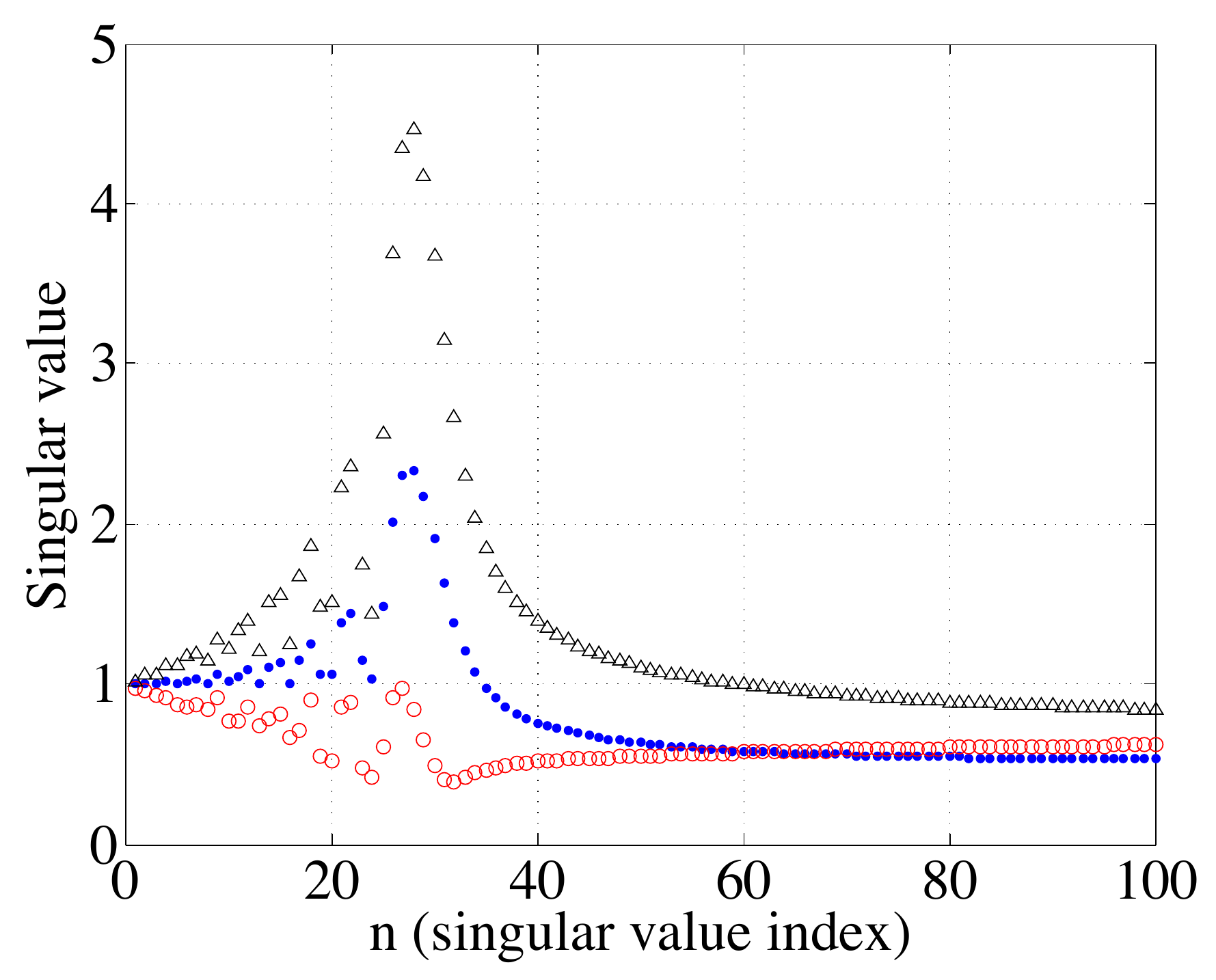}%
\end{center}
\caption{Singular values of the scaled DPIEv integral equation 
(\ref{DPIEv1}) for a spherical scatterer of radius 1 at $k=30$.
Three curves are shown, since (\ref{DPIEv1}) is a vector integral
equation with three unknowns (the scalar $\varrho$ and the tangential surface
vector $\mathbf{a}$).}
\label{fig_svd}
\end{figure}

In Fig. \ref{fig_cond_comp}, we compare the condition numbers of the 
DPIEv and the scaled DPIEv.
Equally revealing is to plot the spectrum of the DPIEv and the scaled DPIEv
(Fig. \ref{fig_Eig_comp}).

\begin{figure}[H]
\begin{center}
\includegraphics[width=3in]{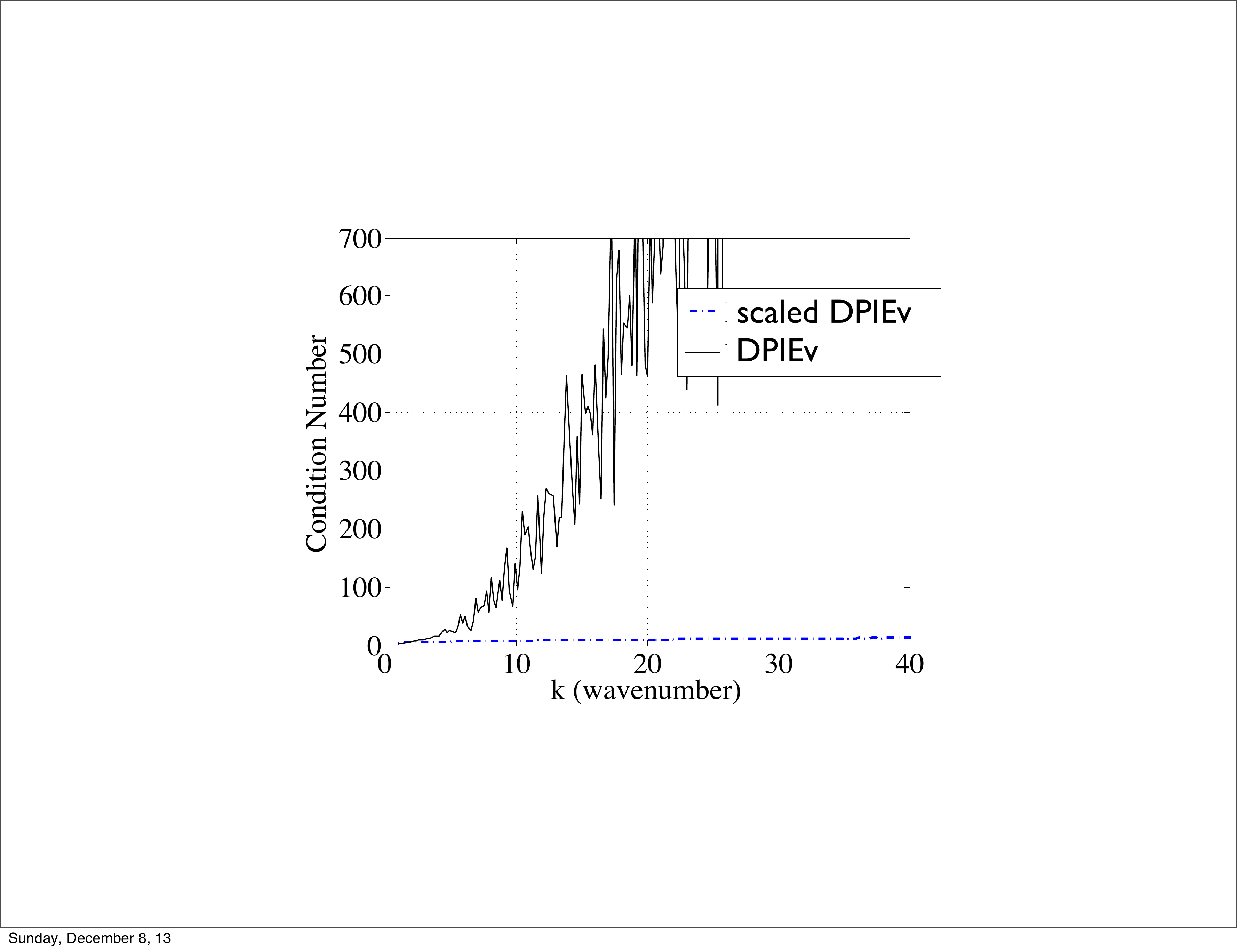}%
\end{center}
\caption{Condition number of the scaled DPIEv (\ref{DPIEv1}) and the
original DPIEv (\ref{HVS1}) for a spherical scatterer of radius 1
as a function of wavenumber $k$.}
\label{fig_cond_comp}
\end{figure}

\begin{figure}[H]
\begin{center}
\includegraphics[width=3in]{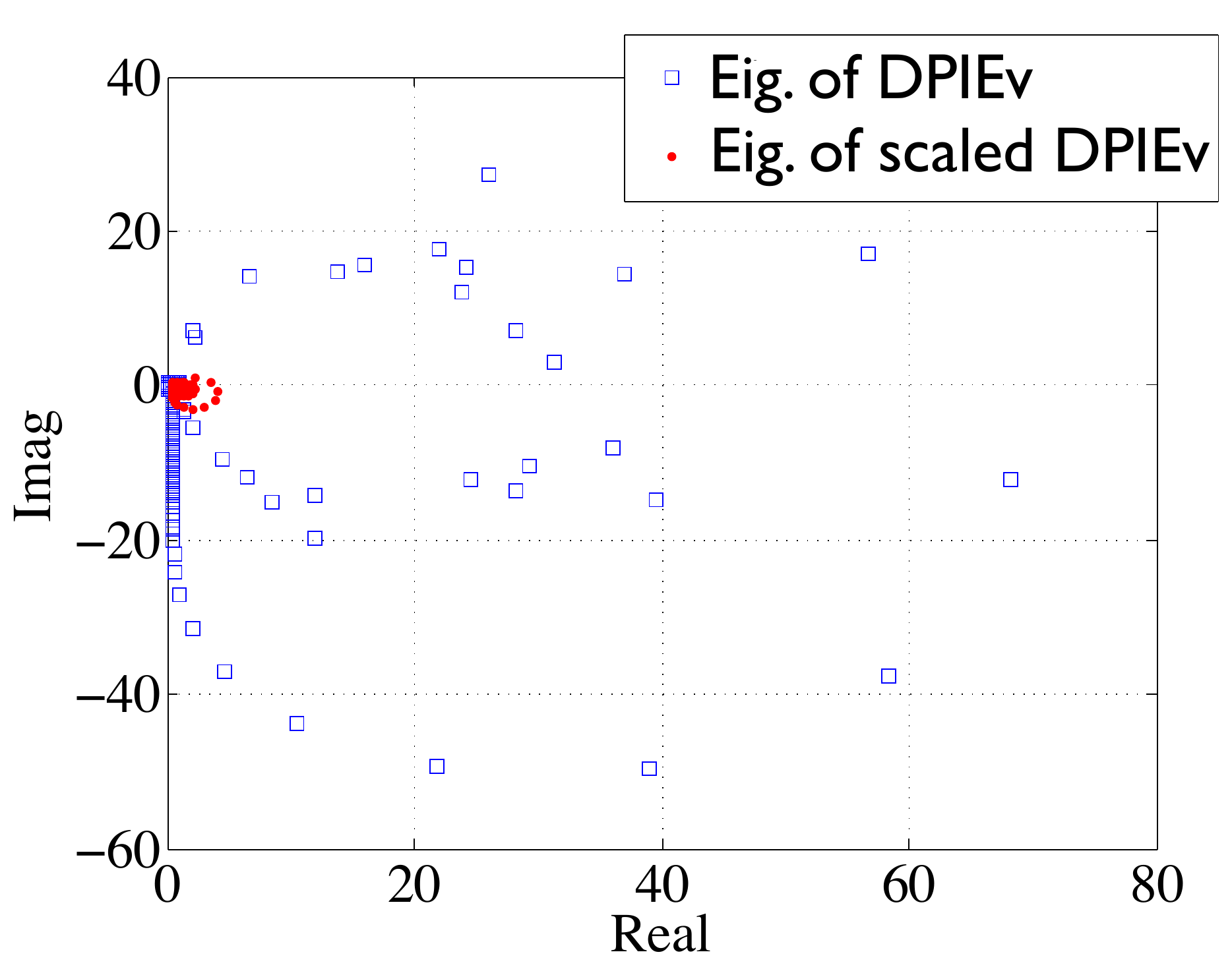}%
\end{center}
\caption{The spectrum of the original DPIEv and the scaled DPIEv (\ref{DPIEv1}) for a 
spherical scatterer of radius 1 at $k=10$.}
\label{fig_Eig_comp}
\end{figure}

While the analysis above has been carried out only for a spherical scatterer,
it is reasonable to expect that the scaled DPIEv is likely to have a big impact even 
for surfaces of arbitrary shape.
As a final point of comparison, 
Fig. \ref{fig_cond_comp2} plots the condition number of various
integral equations that have been suggested for the solution of the Maxwell equations
over a wide range of frequencies. Note that
the scaled DPIE produces slightly worse condition numbers for $1\ll k $. 
We suspect that improvements in the DPIE representation may lead to better
performance, but leave that for future research.
Note, however,
that the scaled DPIE is surprisingly effective in the high-frequency 
regime, despite the fact that it was conceived in order to overcome
topological low-frequency breakdown for scatterers of non-zero genus.

\begin{figure}[H]
\begin{center}
\includegraphics[width=3in]{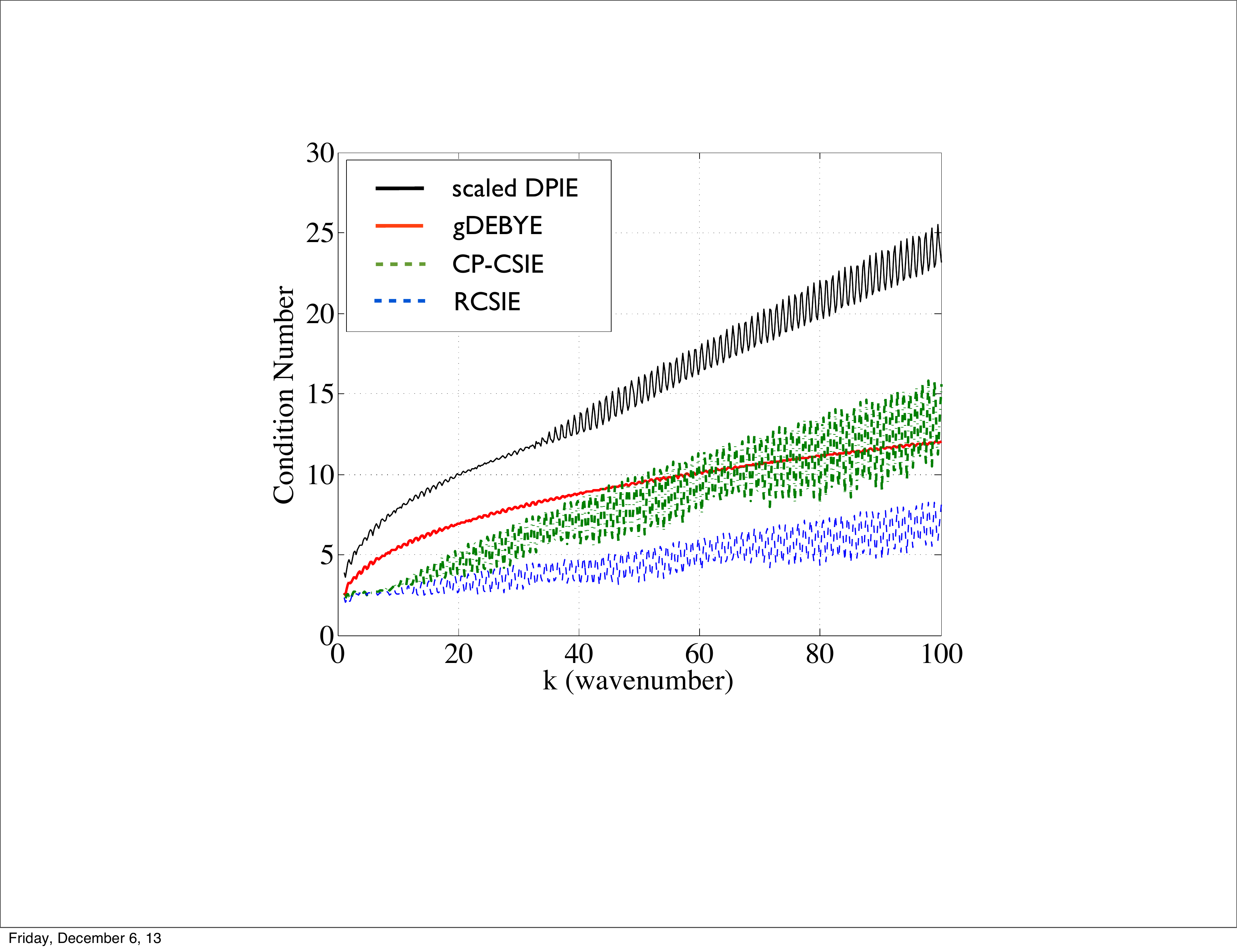}%
\end{center}
\caption{Comparison of the condition numbers of 
several resonance-free integral equations: the
scaled DPIE, the 
generalized Debye source equation (gDEBYE) \cite{gDEBYE}, 
the Calderon preconditioning combined source integral equation (CP-CSIE)
\cite{Rokhlin_Maxwell}, and
the regularized combined source integral equation (RCSIE) \cite{Bruno_Krylov}.}
\label{fig_cond_comp2}
\end{figure}


\section{Conclusions}

We have presented a new formulation of the problem of electromagnetic
scattering from perfect electric conductors. Rather than imposing
boundary conditions on the field quantities themselves, we have
derived well-posed boundary value problems for the vector and scalar
potentials themselves, in the Lorenz gauge. This requires that we
describe incoming fields in the same gauge, but that poses no
fundamental obstacle. We have explained, in section
\ref{incoming:sec}, how to do this for partial wave expansions, plane
waves, and (of course) the potentials induced by known impressed
currents and charges.  We have also developed integral representations
for the vector and scalar potentials that lead to well-conditioned
integral equations (the decoupled potential integral equations or
DPIE). Most importantly, we have shown that the DPIE is insensitive to
the genus of the scatterer. This is one of the few schemes of which we
are aware that does not suffer from {\em topological} low-frequency
breakdown without substantial complications (including the computation
of special basis functions that span the space of surface harmonic
vector fields \cite{gDEBYE}).

Careful analysis of scattering from a unit sphere has demonstrated
that the method works well across a range of frequencies, but the DPIE
is likely to be of particular utility in the low frequency regime ---
especially for structures with complicated multiply-connected
geometry. We will report detailed numerical experiments at a later
date.

\section*{Acknowledgements}
This work was supported in part by the Applied
Mathematical Sciences Program of the U.S. Department of Energy
under Contract DEFGO288ER25053 (L.G.) and
by the Office of the Assistant Secretary of Defense for Research and Engineering 
and AFOSR under NSSEFF Program Award FA9550-10-1-0180 (L.G. and Z. G.) and in part by the
 Spanish Ministry of Science and Innovation (Ministerio de Ciencia e Innovacion) under
 the projects CSD2008-00068 and TEC2010-20841-C04-01.
The authors thank A. Kl\"{o}ckner and M. OÕNeil for many useful discussions.

\appendix

\section{Proofs of existence and stability theorems}\label{theoremsApp}

\vspace{.2in}

\noindent
\textbf{Proof of Theorem \ref{HSSth}}: \\
\noindent
Consider a solution $\tilde{\sigma},\{\tilde{V}_j\}_{j=1}^N$ of the homogeneous 
equation (\ref{HSS1}):
\begin{equation}
\begin{aligned}
\frac{\tilde{\sigma}}{2}+D_k\tilde{\sigma}-i\eta S_k\tilde{\sigma}-\sum_{j=1}^N\tilde{V}_j\chi_j=0,\\
\int_{\partial D_j} \big( D'_k\tilde{\sigma}+i\eta\frac{\tilde{\sigma}}{2} -i\eta S'_k\tilde{\sigma} \big) ds=0 \, .
\end{aligned}
\end{equation}
For this solution, the scalar function and constants
\begin{equation}
\phi^{\Sc}(\bx)=D_k[\tilde{\sigma}](\bx)-i\eta S_k[\tilde{\sigma}](\bx), \ \ \{\tilde{V}_j\}_{j=1}^N\\
\end{equation}
satisfy the scalar modified Dirichlet problem with right-hand side 
$f=0,\{Q_j=0\}_{j=1}^N$. 
By Theorem \ref{uniq_S}, we have $\phi^{\Sc}=0,\{\tilde{V}_j=0\}_{j=1}^N$. 
As $\phi^{\Sc}$ is represented by a combination of single and double layers, 
it is known that $\tilde{\sigma}=0$. This proves uniqueness (see \cite{CK1,CK2}).
Note that the operators $D_k$ and $S_k$ defined on $C^{0,\alpha}(\partial D)$ are compact 
(see \cite{CK1}). The rest of the operators in (\ref{HSS1}) are finite rank, 
so that (\ref{HSS}) is a second kind equation when acting on the space 
$C^{0,\alpha}(\partial D)\times \mathbb{C}^N$, where $ \mathbb{C}^N$ is equipped with 
the usual finite-dimensional topology. Note also that, as a function of $k$, the 
operators involved are continuous in the range $k\in[0,k_{\max}]$ for any fixed 
$k_{\max}$. This implies that the operators involved in equation (\ref{HSS1}) are not 
only compact, but collectively compact as well (see \cite{ANSELONE,ATKINSON}).
By the Fredholm theorem, for any right hand side 
$f,\{\tilde{Q}_j\}_{j=1}^N\in C^{0,\alpha}(\partial D)\times \mathbb{C}^N$, 
there exists a solution 
$\sigma,\{\tilde{V}_j\}_{j=1}^N\in C^{0,\alpha}(\partial D)\times \mathbb{C}^N$.

\vspace{.2in}

\noindent
\textbf{Proof of Theorem \ref{HVSth}}: \\
\noindent
Consider a solution 
$\tilde{\mathbf{a}},\tilde{\varrho},\{\tilde{v}_j\}_{j=1}^N$ of the 
homogeneous equation (\ref{HVS}):
\begin{equation}\label{HVS}
\begin{aligned}
\frac{1}{2}
\left(\begin{array}{c}\mathbf{\tilde{a}}\\ \tilde{\varrho} \end{array}\right) +\overline{\overline{L}}\left(\begin{array}{c}\mathbf{\tilde{a}}\\ \tilde{\varrho} \end{array}\right)+i\eta \overline{\overline{R}}\left(\begin{array}{c}\mathbf{\tilde{a}}\\ \tilde{\varrho} \end{array}\right)-\left(\begin{array}{c}0\\ \sum_{j=1}^Nv_j \chi_j \end{array}\right) = \left(\begin{array}{c}\mathbf{0}\\ 0 \end{array}\right), \\
\int_{\partial D_j} \Big( \bn\cdot\nabla \times S_k\mathbf{\tilde{a}}-\bn\cdot S_k(\bn\tilde{\varrho})+i\eta \big(\bn\cdot S_k(\mathbf{n\times \tilde{a}})-\frac{\tilde{\varrho}}{2}+S'_k V_0\tilde{\varrho} \big)  \Big) ds=0 \, .
\end{aligned}
\end{equation}
For this solution, the vector field and constants
\begin{equation}\label{repAscat}
\bA^{\Sc}=\nabla\times S_k[\mathbf{\tilde{a}}](\bx)- S_k[\bn\tilde{\varrho}](\bx)+i\eta \big( S_k[\mathbf{n\times \tilde{a}}](\bx)+ \nabla S_k[\tilde{\varrho}](\bx) \big), \ \ \{\tilde{v}_j\}_{j=1}^N
\end{equation}
satisfy the vector modified Dirichlet problem, 
with right hand side $\mathbf{f}=0,h=0,\{q_j=0\}_{j=1}^N$. 
By Theorem \ref{uniq_V}, we have $\bA^{\Sc}=0,\{\tilde{V}_j=0\}_{j=1}^N$. 
It is known that a zero field $\bA^{\Sc}=0$ with this representation (\ref{repAscat}) 
must have trivial sources (see \cite{CK1}). Thus, $\mathbf{\tilde{a}}=0$ and 
$\tilde{\varrho}=0$, which proves uniqueness (see \cite{CK2}).

Note that the operator $\overline{\overline{L}}$ defined on 
$T^{0,\alpha}(\partial D)\times C^{0,\alpha}(\partial D)$ is compact 
(see Table \cite{CK1}). 
The operator $\overline{\overline{R}}$ is continuous and bounded on $T^{0,\alpha}(\partial D)\times C^{0,\alpha}(\partial D)$, but not compact. 
However, for a choice of the 
constant $|\eta|<\|\overline{\overline{R}}\|^{-1}$, 
$I+i\eta\overline{\overline{R}}$ has a bounded inverse given by its Neumann series. 
The rest of the operators in equation (\ref{HVS}) are finite rank, so that 
equation (\ref{HVS}) is second kind when acting on the space 
$T^{0,\alpha}(\partial D)\times C^{0,\alpha}(\partial D)\times \mathbb{C}^N$, 
where $ \mathbb{C}^N$ is  equipped with the usual finite-dimensional topology. 
Note also that, as a function of $k$, the operators involved are continuous in the 
range $k\in[0,k_{\max}]$ for any fixed $k_{\max}$. This implies that the operators 
involved in equation \ref{HVS} are not only compact, but collectively compact as well
(see \cite{ANSELONE}).
By Fredholm theory, for any right hand side 
$\mathbf{f},h,\{\tilde{q}_j\}_{j=1}^N\in T^{0,\alpha}(\partial D)\times C^{0,\alpha}(\partial D)\times \mathbb{C}^N$,
there exists a solution 
$\mathbf{a},\varrho,\{\tilde{v}_j\}_{j=1}^N\in T^{0,\alpha}(\partial D)\times 
C^{0,\alpha}(\partial D)\times \mathbb{C}^N$.

\vspace{.2in}

\noindent
\textbf{Proof of Theorem \ref{unif_s}}: \\
\noindent
The solution of the integral equation (\ref{HSS1}) depends continuously on the 
right hand side $f,\{\tilde{Q}_j\}_{j=1}^N$, with the corresponding H\"{o}lder topology. 
Due to the collective compactness of the operators involved in (\ref{HSS1}), 
the continuity is uniform in $k\in[0,k_{\max}]$. 
The map $\phi^{\Sc}(\sigma,\{V_j\}_{j=1}^N)$ is continuous on 
$C^{0,\alpha}(\partial D)\times \mathbb{C}^N\rightarrow  
C^{0,\alpha}(\mathbf{R}^3\backslash D)\times \mathbb{C}^N$ (see \cite{CK1}). 
By composition, the map $\phi^{\Sc}(f,\{Q_j\}_{j=1}^N)$ is continuous as a map
from $ C^{0,\alpha}(\partial D)\times \mathbb{C}^N\rightarrow 
 C^{0,\alpha}(\mathbf{R}^3\backslash D)\times \mathbb{C}^N$,
uniformly on $k\in[0,k_{\max}]$. That is,
\begin{equation}
\| \phi^{\Sc}\|_{0,\alpha, \mathbf{R}^3\backslash D}\le K_{(\alpha,\partial D,k_{max})}\Big(\| f\|_{0,\alpha,\partial D} +\sum_{j=1}^N|Q_j|^2\Big) \, ,
\end{equation}
where the constant $K_{(\alpha,\partial D,k_{max})}$ depends on $\alpha$, 
the surface $\partial D$, and the maximum frequency $k_{max}$. The result is valid 
uniformly down to zero frequency $k=0$.

\vspace{.2in}

\noindent
\textbf{Proof of Theorem \ref{unif_v}}: \\
\noindent
The solution of (\ref{HVS1}) depends continuously on the right hand side with the 
corresponding H\"{o}lder topology. 
Due to the collective compactness of the operators involved in (\ref{HVS1}), the 
continuity is uniform in $k\in[0,k_{\max}]$.  
The map $\bA^{\Sc}(\mathbf{a},\varrho,\{v_j\}_{j=1}^N):
T^{0,\alpha}(\partial D)\times C^{0,\alpha}(\partial D)\times \mathbb{C}^N\rightarrow 
 T^{0,\alpha}(\mathbf{R}^3\backslash D)\times C^{0,\alpha}(\mathbf{R}^3\backslash D)\times 
\mathbb{C}^N$ is continuous (see \cite{CK1}). By composition, the map 
$\bA^{\Sc}(\mathbf{f},h,\{q_j\}_{j=1}^N): 
T^{0,\alpha}(\partial D)\times C^{0,\alpha}(\partial D)\times \mathbb{C}^N\rightarrow 
 T^{0,\alpha}(\mathbf{R}^3\backslash D)\times C^{0,\alpha}(\mathbf{R}^3\backslash D)\times 
\mathbb{C}^N$ is continuous, uniformly in $k$ for $k\in[0,k_{\max}]$. That is,
\begin{equation}
\| \bA^{\Sc}\|_{0,\alpha, \mathbf{R}^3\backslash D}\le K_{(\alpha,\partial D,k_{max})}\Big(\| \mathbf{f}\|_{0,\alpha,\partial D}+\| h\|_{0,\alpha,\partial D} +\sum_{j=1}^N|q_j|^2\Big) \, ,
\end{equation}
where the constant $K_{(\alpha,\partial D,k_{max})}$ depends on $\alpha$, the surface $\partial D$ and the maximum frequency $k_{max}$. The result is valid uniformly down to zero 
frequency $k=0$.

\vspace{.2in}

\section{Partial wave expansions}\label{Partial_wave_expansion}

An important representation of the 
electromagnetic field is that based on separation of variables 
in spherical coordinates. As shown independently by Lorenz, Debye and Mie
\cite{gDEBYE,PAPAS}, the fields induced by sources in the interior of a sphere
can always be expressed in the exterior of the sphere according to the representation:
\begin{equation}\label{Mie1}
\begin{aligned}
\bE^{far}&=\sum_{m,n} \big[a_{mn} \nabla \times \nabla \times (\bx h_n(k|\bx|)Y_n^m)+i\omega\mu b_{mn} \nabla \times (\bx  f_n(k|\bx|)Y_n^m)\big],\\
\bH^{far}&=\sum_{m,n} \big[b_{mn} \nabla \times \nabla \times (\bx h_n(k|\bx|)Y_n^m)-i\omega\epsilon a_{mn} \nabla \times (\bx f_n(k|\bx|)Y_n^m)\big] ,
\end{aligned}
\end{equation}
where $h_n$ is the spherical Hankel function of the first kind.
For sources in the exterior of the sphere, we have 
\begin{equation}\label{Mie1a}
\begin{aligned}
\bE^{loc}&=\sum_{m,n} \big[a_{mn} \nabla \times \nabla \times (\bx j_n(k|\bx|)Y_n^m)+i\omega\mu b_{mn} \nabla \times (\bx  f_n(k|\bx|)Y_n^m)\big],\\
\bH^{loc}&=\sum_{m,n} \big[b_{mn} \nabla \times \nabla \times (\bx j_n(k|\bx|)Y_n^m)-i\omega\epsilon a_{mn} \nabla \times (\bx f_n(k|\bx|)Y_n^m)\big] ,
\end{aligned}
\end{equation}
where $j_n$ is the spherical Bessel function \cite{Abram}.
In order to obtain a finite static limit, we renormalize and define the 
modified spherical Hankel/Bessel function by 
\begin{equation}
\label{fn_norm}
\begin{aligned}
\widetilde{f}_n(k,r):=
\left\{\begin{array}{ll}
\widetilde{h}_n(k,r)=h_n(kr)\frac{k^{n+1}}{-i(2n-1)(2n-3)...5\cdot3\cdot1},\\
\widetilde{j}_n(k,r)=j_n(kr)\frac{(2n+1)(2n-1)...5\cdot3\cdot1}{k^{n+1}}.
\end{array}\right.
\end{aligned}
\end{equation}
It is easy to check that 
\begin{equation}
\begin{aligned}
\lim_{k\rightarrow 0}\widetilde{f}_n(k,r)=
\left\{\begin{array}{ll}
\lim_{k\rightarrow 0}\widetilde{h}_n(k,r)=\frac{1}{r^{n+1}},\\
\lim_{k\rightarrow 0}\widetilde{j}_n(k,r)=r^n.
\end{array}\right.
\end{aligned}
\end{equation}

With a slight abuse of notation we will refer to both $\bE^{far}$ and
$\bE^{loc}$ as $\bE^{\In}$, and to both $\bH^{far}$ and
$\bH^{loc}$ as $\bH^{\In}$. When the distinction is important, we will specify
the use of $\widetilde{h}_n(k,r)$ or 
$\widetilde{j}_n(k,r)$ as the radial function of interest.

Normalizing the coefficients $a_{mn}, b_{mn}$ by the inverse of the scaling factor in
(\ref{fn_norm}),
we write:
\begin{equation}
\begin{aligned}
\bE^{\In}&=\sum_{m,n} \big[\widetilde{a}_{mn} \nabla \times \nabla \times (\bx \widetilde{f}_n(k,|\bx|)Y_n^m)+i\omega\mu\widetilde{b}_{mn} \nabla \times (\bx  \widetilde{f}_n(k,|\bx|)Y_n^m)\big],\\
\bH^{\In}&=\sum_{m,n} \big[\widetilde{b}_{mn} \nabla \times \nabla \times (\bx \widetilde{f}_n(k,|\bx|)Y_n^m)-i\omega\epsilon\widetilde{a}_{mn} \nabla \times (\bx  \widetilde{f}_n(k,|\bx|)Y_n^m)\big] \, .
\end{aligned}
\end{equation}
The fields of a magnetic multipole of degree $n$ and order $m$ are
defined to be
\begin{equation}
\begin{aligned}
\dot{\bE}_{nm}^{\In}&=i\omega\mu\nabla \times (\bx  \widetilde{f}_n(k,|\bx|)Y_n^m), \\
\dot{\bH}_{nm}^{\In}&=\nabla \times \nabla \times (\bx  \widetilde{f}_n(k,|\bx|)Y_n^m) \, .
\end{aligned}
\end{equation}
The corresponding vector and scalar potentials can be defined by
\begin{equation}
\begin{aligned}
\dot{\bA}_{nm}^{\In}&=\mu\nabla \times (\bx  \widetilde{f}_n(k,|\bx|)Y_n^m),\\
\dot{\phi}_{nm}^{\In}&=0 \, .
\end{aligned}
\end{equation}
They clearly satisfy
\begin{equation}
\begin{aligned}
\Delta \dot{\phi}_{nm}^{\In}+k^2\dot{\phi}_{nm}^{\In}=0,\\
\Delta \dot{\bA}_{nm}^{\In}+k^2\dot{\bA}_{nm}^{\In}=0,\\
\nabla\cdot\dot{\bA}_{nm}^{\In}=i\omega\mu\epsilon\dot{\phi}_{nm}^{\In}  \, .
\end{aligned}
\end{equation}
Moreover,
$\dot{\bA}_{mn}^{\In}$ and $\dot{\phi}_{mn}^{\In}$ are bounded.
The fields of an electric multipole of degree $n$ and order $m$ are defined to be
\begin{equation}
\begin{aligned}
\ddot{\bE}_{nm}^{\In}&=\nabla \times \nabla \times (\bx  \widetilde{f}_n(k,|\bx|)Y_n^m),\\
\ddot{\bH}_{nm}^{\In}&=-i\omega\epsilon\nabla \times (\bx  \widetilde{f}_n(k,|\bx|)Y_n^m) \, .
\end{aligned}
\end{equation}
In this case, however, it is easy to verify that the function which serves as the 
obvious vector potential, namely $\bx  \widetilde{f}_n(k,|\bx|)Y_n^m$,
is {\em not} in the Lorenz gauge. To find a suitable replacement,
we compute:
\begin{equation}\label{miedebye}
\begin{aligned}
\nabla \times \nabla \times (\bx  \widetilde{f}_n(k,|\bx|)Y_n^m)=k^2\bx  \widetilde{f}_n(k,|\bx|)Y_n^m+\nabla\frac{\partial}{\partial r}\big(r\widetilde{f}_n(k,r)Y_n^m\big)_{r=|\bx|} \, .
\end{aligned}
\end{equation}
Note that 
\begin{equation}\label{myident}
\begin{aligned}
\frac{\partial}{\partial r}\big(r\widetilde{f}_n(k,r)Y_n^m\big)=\widetilde{f}_n(k,r)Y_n^m+r\frac{\partial}{\partial r}\widetilde{f}_n(k,r)Y_n^m \, .
\end{aligned}
\end{equation}
The first term $\widetilde{f}_n(k,r)Y_n^m$ is a Helmholtz potential. 
Making use of the following identity for spherical Hankel and Bessel functions \cite{Abram}
\begin{equation}\label{ident_abram}
\begin{aligned}
\frac{n+1}{z}h_n(z)+h'_n(z)=h_{n-1}(z),  \\
\frac{n}{z}j_n(z)-j'_n(z)=j_{n+1}(z),
\end{aligned}
\end{equation}
we have
\begin{equation}
\begin{aligned}
r\frac{\partial }{\partial r}\widetilde{h}_n(k,r)=\frac{rk^2}{2n-1}\widetilde{h}_{n-1}(k,r)-(n+1)\widetilde{h}_n(k,r) \, 
\end{aligned}
\end{equation}
and
\begin{equation}
\begin{aligned}
r\frac{\partial }{\partial r}\widetilde{j}_n(k,r)=-\frac{rk^2}{2n+3}\widetilde{j}_{n+1}(k,r)+n\widetilde{j}_n(k,r) \, .
\end{aligned}
\end{equation}
Multiplying by $Y_n^m$,
\begin{equation}\label{myident2}
\begin{aligned}
r\frac{\partial }{\partial r}\widetilde{h}_n(k,r)Y_n^m&=\frac{rk^2}{2n-1}\widetilde{h}_{n-1}(k,r)Y_n^m-(n+1)\widetilde{h}_n(k,r)Y_n^m, \\
r\frac{\partial }{\partial r}\widetilde{j}_n(k,r)Y_n^m&=-\frac{rk^2}{2n+3}\widetilde{j}_{n+1}(k,r)Y_n^m+n\widetilde{j}_n(k,r)Y_n^m.
\end{aligned}
\end{equation}
The first term on the right-hand side is of the order $O(k)$, while
the second term is a Helmholtz potential and of magnitude $O(1)$. 
Using (\ref{myident2}) and (\ref{myident}) in (\ref{miedebye}), we obtain
\begin{equation}
\begin{aligned}
\nabla \times \nabla \times (\bx  \widetilde{h}_n(k,|\bx|)Y_n^m)=&\ k^2\bx  \widetilde{h}_n(k,|\bx|)Y_n^m+\frac{k^2}{2n-1}\nabla\big(r\widetilde{h}_{n-1}(k,r)Y_n^m\big)_{r=|\bx|}\\
\\& -\nabla\big(n\widetilde{h}_n(k,r)Y_n^m\big)_{r=|\bx|} \, ,
\end{aligned}
\end{equation}
for the outgoing waves, and
\begin{equation}\label{myident3}
\begin{aligned}
\nabla \times \nabla \times (\bx  \widetilde{j}_n(k,|\bx|)Y_n^m)=&\ k^2\bx  \widetilde{j}_n(k,|\bx|)Y_n^m-\frac{k^2}{2n+3}\nabla\big(r\widetilde{j}_{n+1}(k,r)Y_n^m\big)_{r=|\bx|}\\
\\& +\nabla\big((n+1)\widetilde{j}_n(k,r)Y_n^m\big)_{r=|\bx|} \, ,
\end{aligned}
\end{equation}
for the incoming waves.
Note that the last term on the right-hand side of (\ref{myident3}) and the 
left-hand side of (\ref{myident3}) both
satisfy the vector Helmholtz equation. Thus, the first two terms on the right-hand side
of (\ref{myident3}) must together satisfy the vector Helmholtz equation as well.
Dividing those two terms by $k^2$ and multiplying by $-i\omega \mu\epsilon$, we define the corresponding vector and scalar potentials by
\begin{equation}
\begin{aligned}
\ddot{\bA}_{nm}^{\In}&=-i\omega\mu\epsilon\bx  \widetilde{h}_n(k,|\bx|)Y_n^m+\frac{-i\omega\mu\epsilon}{2n-1}\nabla\big(r\widetilde{h}_{n-1}(k,r)Y_n^m\big)_{r=|\bx|},\\
\ddot{\phi}_{nm}^{\In}&=n\widetilde{h}_n(k,|\bx|)Y_n^m \,  ,
\end{aligned}
\end{equation}
for the outgoing waves, and
\begin{equation}
\begin{aligned}
\ddot{\bA}_{nm}^{\In}&=-i\omega\mu\epsilon\bx  \widetilde{j}_n(k,|\bx|)Y_n^m+\frac{i\omega\mu\epsilon}{2n+3}\nabla\big(r\widetilde{j}_{n-1}(k,r)Y_n^m\big)_{r=|\bx|},\\
\ddot{\phi}_{nm}^{\In}&=-(n+1)\widetilde{j}_n(k,|\bx|)Y_n^m \, ,
\end{aligned}
\end{equation}
for the incoming waves.
It is easy to verify that
\begin{equation}
\begin{aligned}
\Delta \ddot{\phi}_{nm}^{\In}+k^2\ddot{\phi}_{nm}^{\In}=0,\\
\Delta \ddot{\bA}_{nm}^{\In}+k^2\ddot{\bA}_{nm}^{\In}=0,\\
\nabla\cdot\ddot{\bA}_{nm}^{\In}=i\omega\mu\epsilon\ddot{\phi}_{nm}^{\In} \, .
\end{aligned}
\end{equation}
The last equation, which enforces the Lorenz gauge, is obtained by 
taking the divergence of (\ref{myident3}).
Clearly, both potentials $\ddot{\bA}_{mn}^{\In}$ and $\ddot{\phi}_{mn}^{\In}$ 
are of the order $O(1)$.

\newpage \bibliographystyle{unsrt}

\end{document}